\documentclass[11pt,letterpaper]{article}
\usepackage[margin=1in]{geometry}
\usepackage{times}

\usepackage{float}
\usepackage{amsfonts}
\usepackage{amsmath,amssymb}

\usepackage{amstext}
\usepackage{theorem}
\usepackage{graphicx}
\usepackage{url}
\usepackage{graphics}
\usepackage{colordvi}
\usepackage{subfig}

\usepackage{caption}

\usepackage{latexsym}
\usepackage{epic}
\usepackage{epsfig}
\usepackage{hyperref}
\usepackage{multirow}
\usepackage{hhline}
\usepackage{verbatim}
\usepackage{caption}
\usepackage{enumitem}
\usepackage{color}
\allowdisplaybreaks[1]
\usepackage[numbers]{natbib}
\usepackage{comment}
\usepackage{lipsum}
\usepackage{booktabs} 

\usepackage[linesnumbered,ruled]{algorithm2e} 
\SetKwComment{Comment}{/* }{ */}

        {\hspace*{\fill}$\Box$\par\vspace{4mm}}

\renewcommand{\P}{\mbox{\sf P}}
\newcommand{\NP}{\mbox{\sf NP}}

\newcommand{\opt}{\mathsf{OPT}}

\newcommand{\be}{\begin{enumerate}}
\newcommand{\ee}{\end{enumerate}}
\newcommand{\bd}{\begin{description}}
\newcommand{\ed}{\end{description}}
\newcommand{\bi}{\begin{itemize}}
\newcommand{\ei}{\end{itemize}}

\newtheorem{theorem}{Theorem}[section]
\newtheorem{lemma}[theorem]{Lemma}

\newtheorem{claim}[theorem]{Claim}
\newtheorem{proposition}[theorem]{Proposition}
\newtheorem{assumption}{Assumption}[section]
\newtheorem{definition}{Definition}[section]

\newtheorem{remark}{Remark}[section]

\newenvironment{proof}{\par \smallskip{\bf Proof: }}{\hfill\stopproof}
\def\stopproof{\square}
\def\square{\vbox{\hrule height.2pt\hbox{\vrule width.2pt height5pt \kern5pt
\vrule width.2pt} \hrule height.2pt}}

\newcommand{\eps}{\epsilon}

\newcommand{\poly}{\operatorname{poly}}

\newcommand{\E}{\ensuremath{\mathbb E}}

\newcommand{\prob}[2][]{\text{\bf Pr}_{#1}\left (#2\right)}

\mathchardef\hyphen="2D

\usepackage{xcolor}

\usepackage{thm-restate}

\newcommand{\amb}{\textsf{AMB}}
\newcommand{\menu}{\textsf{EXP-menu}}
\newcommand{\single}{\textsf{EXP-single}}

\begin{document}
\title{Hiring for An Uncertain Task: \\Joint Design of Information and Contracts\thanks{Authors are listed in $\alpha$-$\beta$ order.}} 
\author{
Matteo Castiglioni \\ Politecnico di Milano \\ \small matteo.castiglioni@polimi.it
\and
Junjie Chen\thanks{Work was done while Junjie Chen was a visiting student at Osaka University.} \\ City University of Hong Kong \\ \small junjchen9-c@my.cityu.edu.hk
}
\date{}

\begin{titlepage}
	\clearpage\maketitle
	\thispagestyle{empty}

\begin{abstract}

Information design and contract design are two important topics in microeconomics. In a principal-agent model, information design concerns the information asymmetry between two players, while contract design addresses the moral hazard challenge. Though they are typically studied as two disjoint problems, in many real-world scenarios, e.g., recruitment, information asymmetry and moral hazard could appear simultaneously. 

In this paper, we initiate the computational problem of jointly designing information and contracts. 
We consider three possible classes of contracts with decreasing flexibility and increasing simplicity: ambiguous contracts, menus of explicit contracts and explicit single contract. Ambiguous contracts allow the principal to conceal the applied payment schemes through a contract that depends on the unknown state of nature, while explicit contracts reveal the contract prior to the agent's decision.
Our results show a trade-off between the simplicity of the contracts and the computational complexity of the joint design.
Indeed, we show that an approximately-optimal mechanism with ambiguous contracts can be computed in polynomial time.
However, they are convoluted mechanisms and not well-suited for some 
real-world scenarios.
Conversely, explicit menus of contracts and single contracts are simpler mechanisms, but they cannot be computed efficiently.
In particular, we show that computing the optimal mechanism with explicit menus of contracts and single contracts is APX-Hard.
We also characterize the structure of optimal mechanisms.
Interestingly, direct mechanisms are optimal for both the most flexible ambiguous contracts and the least flexible explicit single contract, but they are suboptimal for that with menus of contracts.
Finally, motivated by our hardness results,  we turn our attention to menus of linear contracts and single linear contracts. We show that both the problem of computing the optimal mechanism with an explicit menu of linear contracts and an explicit single linear contract admits an FPTAS.

\end{abstract}

\end{titlepage}

\section{Introduction}


Information design and contract design both play fundamental roles in microeconomics, each addressing different challenges. Information design is applied when there is {\it information asymmetry} between the principal and the agent, allowing the principal to influence agent's decisions \cite{kamenica2011bayesian,arieli2019private}. 
Information asymmetry is common in practice. One typical example is online ads auctions~\cite{badanidiyuru2018targeting,bergemann2021calibrated,emek2014signaling,chen2023bayesian,bro2012send}. Online platforms have access to users' cookies and therefore know more about the users than advertisers. By strategically revealing information about the users or click-through rates about displayed ads, the platforms can influence the advertisers' bidding behaviors to align with the principal’s objectives. 
In comparison, contract design addresses the  challenge of  {\it moral hazard},  where the principal cannot observe the agent's action but can observe certain outcomes resulting from that action~\cite{bolton2004contract,carroll2015robustness,holmstrom1980theory}. Hence, to incentivize the agent to take specific actions, the principal offers contracts that pays the agents on the basis of the realized outcomes. 
Given their great success, e.g., Nobel Prize 2016 awarded to the field of contract design \cite{nobel-contract}, and their practical significance, it is unsurprising that both fields have recently received significant interest from the theoretical computer science community~\cite{xu2020tractability,dughmi2017algorithmic,dughmi2017algorithmicnoex,dughmi2019algorithmic,dutting2019simple,duetting2024multi,dutting2021complexity,castiglioni2021bayesian,castiglioni2024reduction}. 

In the classical information design model, the principal can observe the state of nature $\theta \in \Theta$ hidden to the agent. The principal designs and commits to a signaling scheme $\pi: \Theta \to \Delta_\Sigma $, mapping each state to a distribution over a set of signals $\Sigma$. The goal of the principal is to incentivize the agent to play favorable actions. 
Similarly, in contract design the principal's goal is to design a contract $p$, also known as payment scheme,  to incentivize an agent.
The agent has $n$ possible actions and each action induces a probability distribution over $m$ possible outcomes. In particular, a matrix $F\in \mathbb{R}_+^{n\times m}$ maps actions to possible outcomes, where $F_{ij}$ is the probability that action $a_i$ leads to outcome $\omega_j$. Then, the contract $p\in \mathbb{R}_+^m$ specifies the payment $p_j\ge 0$ from the principal to the agent when outcome $\omega_j$ is observed. 

In this paper, we aim at combining the two research fields, providing a model that subsumes both the classical information design and contract design problem.
Our goal is to deal with the two challenges together: incentivize agents through information and through payments under moral hazard.
In Section~\ref{sec:model}, we briefly introduce our mode, while in Section~\ref{sec:results} we highilght our results. Finally, in Section~\ref{sec:related} we discuss the more relevant related works. 

\subsection{Model}\label{sec:model}

The problem of jointly designing information and contracts was motivated when one of the authors was watching the manga series {\it One Piece}. 
In this manga world, the government issues bounty notices to incentivize bounty hunters (e.g., the character Zoro) to {pursue criminals}.
Two important pieces of information typically need to be included: a description of criminals and a reward for the hunter who catches them. Clearly, it is critical  to consider the design of these two factors simultaneously to make the notice effective. 
Then, given the (incomplete) information in the criminal's description and the reward, the bounty hunter either makes his plan for the task or quits. 

Many similar scenarios are common in practical life. 
One familiar example is a lost item notice, where the owner needs to provide information about the lost item, e.g., the item's appearance, where it was lost, etc., and promise a reward for those who find the lost item. 
Another example is employee recruitment. To seek qualified individuals for a position, the recruitment posts usually need to contain the job description (e.g., skill requirements and responsibilities of the position) and more importantly the salary which could depend on performance, e.g., base salary  plus bonus. 
The model is also applicable to online car-hailing and hitchhiking services which comprise a billion-dollar market~\cite{statista2024}. Typically, users publish their ride requests (including details such as destination, pick-up time, and number of passengers) and specify the payments to drivers on online platforms, e.g., DiDi in China and Uber in US, so that interested drivers will accept these ride orders.

The essential characteristic of all these examples is that the principal needs to provide information and reward to the agent, i.e., the joint design of information and contracts.
Formally, we model this settings as a game between a principal and an agent, where the principal can 
exclusively observe the state of nature $\theta$. 
Each $\theta$ is associated with a matrix $F^{\theta} \in \mathbb{R}_+^{n\times m}$.
Intuitively, this models that the success of an action depends on the  environment, which is unknown to the agent. To exploit the information asymmetry, the principal designs a signaling scheme $\pi$ to strategically reveal information through signals $s\in \Sigma$, e.g., the value of lost items in the {loss item notice} example.   
At the same time, the principal commits to a payment scheme, e.g., a fraction of the value of the returned item. Conditional on the received signal and the payment scheme, the agent chooses the action more aligned with his interests.

While it is well-known that both the classical information design and the classical contract design can be solve efficiently through linear programming~\cite{dughmi2019algorithmic,dutting2019simple}, the interplay of information and contracts makes the problem intrinsically harder. This might be expected since previous works, e.g., \cite{badanidiyuru2018targeting,emek2014signaling}, adding information design to an simple models can potentially lead to intractable problems. For example, \citet{emek2014signaling} show that applying information design to a second-price auction leads to an \NP-hard problem.
The closest work to ours is~\cite{dughmi2019persuasion}, which studies information design augmented with payments. We remark that in \cite{dughmi2019persuasion}, payment is conditional on {\it observable} actions, and the payments are used to ``reinforce" persuasiveness, which is indeed still addressing only the information asymmetry. In contrast, the payments in our model are designed for {\it unobservable} actions to address moral hazard. Without payments, the agent will choose to quit regardless of  the information. 


\subsection{Our Results}\label{sec:results}

In this paper, our main focus is on the computational aspects of jointly designing information and contracts. We consider three types of contracts: ambiguous contracts, menus of explicit contracts and single explicit contract, which have decreasing flexibility and increasing simplicity. It turns out that the joint information and contract design problem in its natural formulation is a challenging {\it bilinear program} independently from the chosen class of contracts.

In the case of ambiguous contracts, for each signal, the principal commits to a payment scheme in which a different contract is associated to each state. This allows the principal to withhold the contracts and further increases the agent's uncertainty. Interestingly, our first result shows that the joint design problem might not admit a maximum. In contrast to the close joint design problem of information design and auctions~\cite{bergemann2007information}, which is $\NP$-hard~\cite{cai2024algorithmic}, our second result shows that our problem can be approximated to within an arbitrary small factor in polynomial time.   To prove our result, we first show that it is without loss of generality to consider direct mechanisms, i.e., in which the principal directly recommends actions to the agent. This allows us to employ a linear relaxation technique~\cite{gan2022optimal,cacciamani2024multi} to reformulate our problem as a linear program.  By exploring the structure of our problem, we develop a simple algorithm to fix the {\it irregular} cases introduced by linear relaxation. We first slightly modify the scheme $\pi$ with a {\it reserve and redistribute} technique \cite{babichenko2021multi,chen2023bayesian} to obtain a $\xi$-incentive-compatible ($\xi$-IC) mechanism, then we apply the $\xi$-IC to IC transformation \cite{castiglioni2022designing,castiglioni2024reduction,dutting2021complexity} to retain a feasible solution with a small loss in utility.

Regarding explicit contracts, our results are mainly about hardness. 
Our first set of results is devoted to joint design with a menu of explicit contracts.
Our first result is unexpectedly negative: it is suboptimal to restrict to direct mechanisms, even in a simple case with only $2$ actions, $2$ outcomes and $3$ states. This implies that the revelation principle usually applied in the Bayesian persuasion literature~\cite{dughmi2017algorithmicnoex,dughmi2017algorithmic,kamenica2011bayesian} ceases to work, and the number of signals required by an optimal mechanism could potentially be large. We leave the following as an interesting open question:\par
~~~~~~~~~~{\it How many signals are required by an optimal mechanism with a menu of explicit contracts?}\par
Next, we focus on the computational complexity of the problem. The failure of the revelation principle makes the problem more challenging. Indeed, our second results show that it is \NP-hard to find an optimal mechanism with a menu of explicit contracts. Our proof is by reducing from the problem of finding disjoint independent sets in almost $q$-colorable graphs \cite{khot2012hardness}.
The graph has $n$ vertices where each vertex is associated with a delicate matrix $F^{\theta}$ so that a connection is built between the principal's expected utility and the size of an independent set.
Interestingly, in our reduction, optimal contracts require high payments. Our next result shows that this is a necessary condition.
Indeed, we show that an approximately optimal mechanism with \emph{bounded} contracts can be computed in quasi-polynomial time. In particular, if payments are constrained to be bounded by a constant $B>0$, the problem admits a quasi-polynomial-time approximation scheme (QPTAS), i.e., that for any constant $\epsilon>0$ returns a $\epsilon$-approximation in quasi-polynomial time. Our key idea is restrict to a finite set of possible distributions called $K$-uniform distributions~\cite{cheng2015mixture}, which will lead to an approximately optimal and   approximately IC solution. This is then converted into an IC solution through a procedure similar to the one used for ambiguous contracts.

Our second set of results concerns the joint design with a single explicit contract. While direct mechanisms are suboptimal for menus of explicit contracts case, it is surprising that direct mechanisms are optimal for the joint design with a single contract. Moreover, an optimal mechanism, i.e., that achieves the maximum, always exists, which is also in sharp contrast with the ambiguous contract case. In light of this positive result, one may expect this problem to be polynomially solvable. Unfortunately, 
we show that also this problem is APX-Hard. Our proof reduces from the problem of finding the smallest dominating set on a graph with degree at most $3$~\cite{chlebik2008approximation}. The key observation is that if the signaling scheme $\pi$ is full information revelation, our problem reduces to Bayesian contract design with a single contract, which is APX-Hard~\cite{guruganesh2021contracts,castiglioni2021bayesian}. Inspired by this observation,  we construct hard instances in which optimal mechanisms involve a full information revelation strategy.

Given the hardness results of general joint design problems with explicit contracts, we turn our attention to one common class of simple contracts: linear contracts~\cite{carroll2015robustness,alon2022bayesian,dutting2019simple}.  We consider both single linear contracts and menus of linear contracts. In both cases, we present a fully-polynomial-time approximation scheme (FPTAS)  that gives $\eps$-approximation for any $\eps>0$ and runs in polynomial time. Our result is based on the observation that it sufficient to consider only $1/\epsilon$ possible contracts to obtain an $\epsilon$ approximation. One challenge in proving the FPTAS for the case of menus of linear contracts is that the number of signals required is unknown due to the suboptimality of the direct mechanism. We overcome this issue by observing that two signals inducing the same action under the same linear contract can be combined. Hence, $\epsilon$-optimal signaling schemes require at most $\frac{n}{\eps}$ signals.

\subsection{Additional Discussion on Related Works}\label{sec:related}

In this section, we provide more discussion on the related works. Our work is related to contract design, information design and the application of information design in auctions.

{\bf Related Works on Contract Design.} Contract theory is one of the fundamental fields in microeconomics \cite{holmstrom1980theory,bolton2004contract,myerson1982optimal,carroll2015robustness}. Recently, it has started to receive attention from computer science communities, and the literature on algorithmic contract design is growing fast. To the best of our knowledge, \citet{babaioff2006combinatorial} is the first to study the contract design problem computationally. \citet{dutting2019simple} studies approximation guarantees of linear contracts. 
After that, the complexity and the approximation problems of contract design start to receive significant attention \cite{dutting2021complexity}. Along this direction, many works focus on the settings under both adverse selection and moral hazard \cite{guruganesh2021contracts,castiglioni2024reduction,castiglioni2022designing,alon2021contracts,castiglioni2021bayesian,guruganesh2023power}. Specifically, it has been shown that the design of  menus of contracts and single contract are \NP-hard in both multi-parameter and single-parameter settings \cite{castiglioni2022designing,guruganesh2021contracts,castiglioni2021bayesian,castiglioni2024reduction}. Apart from the complexity, other algorithmic problems of contract design also receive much attention, such as combinatorial contracts~\cite{dutting2022combinatorial,dutting2024combinatorial}, multi-agent contracts~\cite{duetting2024multi,castiglioni2023multi} and ambiguous contracts~\cite{dutting2023ambiguous}. More recently, attention was also given to the online learning problem of contracts \cite{chen2024bounded,bacchiocchi2023learning,zhu2023sample}.

{\bf Related Works on Information Design.} Information design concerns how a principal reveals information to influence the agent's decisions. There is a rich literature in one of its sub-fields, Bayesian persuasion, which was first studied by \citet{kamenica2011bayesian} and later extended to private Bayesian persuasion by \citet{arieli2019private} where the principal privately communicates with each agent. A series of algorithmic studies \cite{xu2020tractability,babichenko2021multi,dughmi2019algorithmic,dughmi2017algorithmicnoex} follows in theoretical computer science. Importantly, \citet{babichenko2017algorithmic} and \citet{dughmi2017algorithmicnoex} respectively show that both computing an optimal private and public Bayesian persuasion is indeed \NP-hard. There are also other recent works, such as persuasion with limited resource~\cite{gradwohl2022algorithms}, the computation of approximately optimal schemes~\cite{castiglioni2023public} and information design in concave games~\cite{smolin2022information}. We refer the interested reader to the comprehensive survey by \citet{dughmi2017algorithmic}. Another important subfield  is the sale of information, which was first introduced by \citet{babaioff2012optimal}. Later, \citet{chen2020selling} improve the results by designing a consulting mechanism. \citet{bergemann2018design} characterize the optimal information structure for the binary environment, while \citet{cai2020sell} study the algorithmic aspects for a more general setup. \citet{liu2021optimal} give an optimal threshold mechanism for selling information to a binary-decision buyer while \citet{yang2022selling} studies the information selling for segmenting consumers. A series of algorithmic works in this line study information selling in different setups ~\cite{castiglioni2023selling,bergemann2022selling,chen2022selling}.

{\bf Related Works on Information Design in Auctions.} Typically, the principal has an informational advantage over the agent, so the principal can influence agents' bids in auctions through information design approaches. 
\citet{bro2012send} and 
\citet{emek2014signaling} study the problem of information design in second-price auctions. They show that generally, it is \NP-hard to design an optimal signaling scheme to maximize the principal's revenue. \citet{fu2012ad} study the information design in Myerson's auctions. \citet{badanidiyuru2018targeting} study information design in second-price auctions with re-targeting features.  \citet{bergemann2021calibrated} study information design in click-through auctions with calibration constraints and later \citet{chen2023bayesian} extend to a Bayesian setting. \citet{cai2024algorithmic} and \citet{bergemann2007information} consider the joint design problem of information and auctions.

\section{Preliminaries}

In this section, we formally introduce the problem of jointly designing information and contracts. We consider the setting where there is one principal and one agent. The agent can take  one of the $n$ actions available, i.e., action $a_i \in \mathcal{A}$ with $i \in \langle n \rangle \triangleq \{0, 1, \dots, n-1\}$. Each action $a_i$ is associated with a nonnegative cost $c_i\ge 0$. Each action $a_i$ can result in one of the $m$ outcomes $\omega_j \in \Omega$. Each outcome $\omega$ will lead to the reward $r_{\omega}\ge 0$ for the principal. We assume that the cost $c \in \mathbb{R}_{+}^{n}$  and the reward $r \in \mathbb{R}_{+}^{m}$ are known to both the principal and agent. For easy exposition, we assume cost $c$ and reward $r$ are within the interval $[0, 1]$ unless otherwise specified. All our results generalize as long as they are bounded by some constant $C>0$.

The principal wants to hire an agent to perform some tasks. However, differently from the classical principal-agent problem, the agent is uncertain about the task environment. We model this scenario with Bayesian uncertainty on the environment encoded by the state of nature $\theta \in \Theta$. Each state $\theta$ determines a action-outcome mapping probability matrix $F^{\theta} \in \mathbb{R}_+^{n\times m}$, where $F_i^{\theta} \in \Delta^m$ is a probability distribution over outcomes $\Omega$ resulted by taking action $a_i$. We assume that only the principal can observe the state $\theta$, while the prior distribution $\mu$ over state $\Theta$ is publicly available. Moreover, following the literature \cite{castiglioni2022designing,alon2021contracts,guruganesh2021contracts,castiglioni2024reduction}, we  assume that there is an opt-out action $a_0$ for the agent, which is without loss of generality.
\begin{assumption}\label{assmutpionoptout}
    There is an opt-out action $a_0$ with cost $c_0 = 0$.
\end{assumption}



{\textbf{Timeline of Interactions Between Principal and Agent.}} The timeline of the interaction between the principal and the agent is the following: 
\begin{enumerate}
    \item The principal commits to a signaling scheme $\pi$ and a payment scheme $\mathcal{P}$
    \item The state $\theta$ is realized according to distribution $\mu$.
    \item The principal observes the state $\theta$ and then samples a signal from the signaling scheme $\pi$ as $s \sim \pi(s|\theta)$. 
    \item The agent observe the signal $s$ and plays an action $a_i$ 
    \item  An outcome $\omega \sim F^\theta_i$ is realized and the agent is paid according to $\mathcal{P}$. 
\end{enumerate}

In the following, we describe in detail the joint information and contract design problem.

{\textbf {Information Design to Eliminate Uncertainty.}}  To hire an agent, the principal needs to reveal some information to {\it partially} eliminate the agent's uncertainty about the task. 
We model this procedure through  an {\it information design} approach. 

Formally, the principal designs a signaling scheme $\pi: \Theta \to \Sigma$ which maps the state $\theta \in \Theta$ to a distribution over signals $\Sigma$.  After observing the state $\theta$, 
the principal will sample a signal $s\in \Sigma$ with probability $\pi(s|\theta)$ and communicate it with the agent. 
Once observing signal $s$, the agent updates his posterior belief about the states by the Bayesian rule
\begin{equation}\label{posteriorformaitnlleion}
    \prob{\theta|s} = \frac{\mu(\theta)\pi(s|\theta)}{\prob{s}}
\end{equation}
where $\prob{s} = \sum_{\theta} \mu(\theta)\pi(s|\theta)$ is the probability of sending signal $s$. Note that we do {\it not} focus on direct signaling schemes, i.e., $\mathcal{A} = \Sigma$, as many existing works \cite{kamenica2011bayesian,xu2020tractability,dughmi2019algorithmic,dughmi2017algorithmicnoex}. Indeed, we will show that in our problem direct signaling schemes are \emph{not} optimal in general.



{\textbf{Contract Design to Incentivize an Agent.}} If the agent accepts this task, he may take some action based on his own interests. However, this action is private to the agent, while the principal can only observe the realized outcomes. In order to incentivize the agent, the principal  specifies some payments for the realized outcomes of the task. We model this procedure via the {\it contract design} approach.

We consider three classes of payment scheme $\mathcal{P}$ in which the contract depends on both $\theta$ and $s$, depends only on a signal $s$, or is independent of both $\theta$ and $s$. 

We call the first class of payment scheme as {\it ambiguous contract} and denote the corresponding optimization problem as \amb. This allows the principal to conceal the true contract applied and gives the largest flexibility in incentivizing the agent. Indeed, ambiguous contracts are widely seen in practical life.  One typical example is that some lost item notices write the rewards for finders ambiguously, e.g., ``Your efforts will be greatly appreciated" or ``A generous reward awaits". Another example is in the recruitment process, where the employee may be promised a year-end bonus that could be flexible (i.e., ambiguous), e.g., one to six months' salary. We denote the ambiguous contract as $p^{s, \theta} \in \mathbb{R}^m_+$. Hence, the agent is not sure which contract the principal actually applies to since the state is not observable, but has a posterior belief over it after observing the signal $s$. 
Conditional on $\pi$ and $\mathcal{P}$, once the signal $s$ is observed, the agent takes the {\it best-response} action $a(s) \in \mathcal{A}$ that maximizes his own utility under the posterior over states $\prob{\theta|s}$, 
\begin{align}\label{bset-response}
    a(s) = \arg\max_{a_i, i \in \langle n \rangle} \sum_{\theta} \prob{\theta|s} [\langle F^{\theta}_i, p^{s,\theta} \rangle -c_i]
\end{align}
where the tie-breaking is in favor of the principal. $\langle \cdot, \cdot \rangle$ denotes the product of two vectors.

It is worth noting the difference between the {\it ambiguous contract} in our problem and that in  \cite{dutting2023ambiguous}: while the agent in the two problems is not sure about which contract the principal employs, the ambiguity in our problem comes from the agent's uncertainty about the state whereas in \cite{dutting2023ambiguous}, the principal intentionally commits to a set of contracts. Moreover, different from \cite{dutting2023ambiguous}, our \amb~problem focuses on the interplay of contracts and information.

Next, we introduce the {\it explicit contract}, which are the latter two classes of contracts we consider. {\it Explicit} means that the contract applied is revealed to the agent. In the lost item notice example, the owner may instead explicitly write the amount of the award for the finders, e.g., "One million dollars will be given".
We consider two types of explicit contracts: menus of explicit contracts (depending only on signal $s$) and single explicit contract (independent of $\theta$ and $s$).
In the menus of contracts case, we design one contract for each signal, i.e., $p^s\in \mathbb{R}^m_+$, while there is only one  contract $p\in \mathbb{R}^m_+$ in the single contract case.
Once observing signal $s$, the agent chooses the best-response action $a(s)$  as in (\ref{bset-response}) with slight modification on the contracts employed. We denote these two joint design problems as \menu~and \single, respectively.


{\textbf{The Principal's Problem.}} The principal's objective is to maximize his expected utility by jointly designing information structure $\pi$ and contracts $\mathcal{P} = \{p^{s,\theta}\}_{s\in \Sigma, \theta \in \Theta}$. The problem considering ambiguous contracts is  formulated as follows.
\begin{equation}\tag*{(\amb)}
\label{original_problem_contract_signal}
\begin{aligned}
    \sup_{\pi, \mathcal{P}} \quad & \sum_{\theta \in \Theta} \mu(\theta) \sum_{s \in \Sigma} \pi(s|\theta) \langle F^{\theta}_{a(s)},   r - p^{s,\theta} \rangle  \\
    \textnormal{s.t. } \quad & \sum_{\theta \in \Theta} \mu(\theta)\pi(s|\theta) \Big[ \langle F^{\theta}_{a(s)}, p^{s,\theta} \rangle -c_{a(s)} \Big] \ge \sum_{\theta \in \Theta} \mu(\theta)\pi(s|\theta) \Big[ \langle F^{\theta}_{i},  p^{s,\theta} \rangle -c_{i}\Big], \quad  \forall s \in \Sigma, i \in \langle n \rangle \\
& p^{s,\theta} \ge 0, \quad \forall s \in \Sigma, \theta \in \Theta \\
& \sum_{s \in \Sigma} \pi(s|\theta) = 1, \pi(s|\theta) \ge 0, \quad \forall \theta \in \Theta, s \in \Sigma 
\end{aligned}
\end{equation}

The first constraints ensure that $a(s)$ is the best response action for the agent under signal $s$. Moreover, due to the existence of an opt-out action (Assumption \ref{assmutpionoptout}), it also ensures that the agent receives nonnegative utility, i.e., the individual rationality (IR). Following \cite{alon2021contracts,castiglioni2024reduction,myerson1982optimal}, we call the first constraints as {\it incentive-compatible (IC)} constraints. The second constraint is the {\it limited liability} constraint commonly adopted in the literature \cite{castiglioni2024reduction,alon2021contracts,carroll2015robustness}. The final constraint is to ensure that the designed signaling scheme $\pi$ is a probability distribution.

Adding additional equality constraints to Program \ref{original_problem_contract_signal} leads to a more restrictive problem \menu. Specifically, \menu~defined upon Problem \ref{original_problem_contract_signal} is formulated as,  
\begin{equation}\tag*{(\menu)}
\label{original_problem_contract_menu}
\begin{aligned}
    \sup_{\pi, \mathcal{P}} \quad & \textnormal{Problem}~\ref{original_problem_contract_signal}  \\
    \textnormal{s.t. } \quad & p^{s, \theta} = p^{s, \theta'}, \forall \theta, \theta' \in \Theta
\end{aligned}
\end{equation}

Similarly, we formulate the problem \single~with more equality constraints as,
\begin{equation}\tag*{(\single)}
\label{original_problem_contract_single_exp}
\begin{aligned}
    \sup_{\pi, \mathcal{P}} \quad & \textnormal{Problem}~\ref{original_problem_contract_signal}  \\
    \textnormal{s.t. } \quad & p^{s, \theta} = p^{s', \theta'}, \forall \theta, \theta' \in \Theta, \forall s, s' \in \Sigma
\end{aligned}
\end{equation}
It is clear that the problems with the three different types of contracts  exhibit decreasing flexibility. 

\begin{remark}
    In the above three problems \ref{original_problem_contract_signal}, \ref{original_problem_contract_menu} and \ref{original_problem_contract_single_exp}, fixing either the signaling scheme $\pi$ or the payments scheme $\mathcal{P}$ makes the problem an easily solved linear program. The challenges come from the interplay of $\pi$ and $\mathcal{P}$, which makes the problem a constrained bilinear program.
\end{remark}





Before moving to the main results, we formally define the direct mechanism and $\eps$-IC, which are common in the literature~\cite{kamenica2011bayesian,castiglioni2024reduction,castiglioni2022designing,xu2020tractability,arieli2019private,dughmi2017algorithmicnoex}.

\begin{definition}
    (Direct Mechanism) A direct mechanism is a tuple $(\pi, \mathcal{P}, \Sigma)$ such that 1) The size of the signal set is at most $|\Sigma| \le n$;
    2) Under each signal $s \in \Sigma$, the agent takes the best-response action $a$ under contracts $\mathcal{P}$ and no two signals $s, s'\in \Sigma$ induce the same best-response action for the agent.\\ 
    In other words, the principal directly recommends an action to the agent, which is indeed the agent's best-response action. 
\end{definition}

\begin{definition}
    ($\eps$-IC) Under signal s with $\{\pi(s|\theta)\}_{\theta\in \Theta}$ , the action-contract pair $(a, p)$ is $\eps$-IC in an instance \single~if the following holds for $\eps \ge 0$ and all action $a_i, i\in \langle n \rangle$,
    \[
    \sum_{\theta \in \Theta} \mu(\theta)\pi(s|\theta) \Big[ \langle F^{\theta}_{a}, p \rangle -c_{a} \Big] \ge \sum_{\theta \in \Theta} \mu(\theta)\pi(s|\theta) \Big[ \langle F^{\theta}_{i},  p \rangle -c_{i}\Big] -\eps
    \]
    $\eps$-IC is also defined similarly for instances \amb~and \menu.
\end{definition}

\section{Ambiguous Contracts} \label{sec:ambigous}

In this section, we focus on the most flexible setting with ambiguous contracts, \emph{i.e.}, the \amb~ problem. 
Our first result shows that it is without loss of generality to restrict to a direct and IC mechanism, where the principal employs signals that correspond to actions recommendations and the agent obediently follows the recommendations. Notice that this does not implies that there exists an optimal direct and IC mechanism, but just that direct mechanism are as powerful as indirect ones. Formally,

\begin{lemma}\label{lemma_recommenddirect}
Given a feasible mechanism to Problem \ref{original_problem_contract_signal}, there exists a direct and IC mechanism that achieves at least the same principal's utility.
\end{lemma}

We remark that restricting to direct mechanisms is {\it not} without loss of generality in more restricted versions of our problem, as we show in Section \ref{sec:explicitycontractdesing} for the case of explicit contracts. 
Next, armed with Lemma \ref{lemma_recommenddirect}, we shows that surprisingly an optimal (direct and IC) solution to Problem \ref{original_problem_contract_signal} might not exist.

\begin{proposition}
    \label{theoremsupereconmanntoacdhi}
    There exist instances where the maximum of Problem \ref{original_problem_contract_signal} is not achievable, i.e., only the supremum of Problem \ref{original_problem_contract_signal} exists.
\end{proposition}

We prove this result by constructing a simple instance with $3$ states, $3$ actions and $4$ outcomes. While the expected payment from the principal to the agent could be bounded, the non-achievable maximum exists due to situations where the principal may send some signal with sufficiently small probability but induce a very large payment on some outcome in the contract associated with that signal.


Now, we show that even if an optimal mechanism might not exist, it is possible to find an arbitrarily good approximation of the optimal mechanism in polynomial time.
We remark that as Proposition~\ref{theoremsupereconmanntoacdhi} shows, the approximation is not due to computational reasons but the non-existence of the maximum. 

\begin{theorem}
    For any $\zeta>0$, there exists one (additively) $\zeta$-optimal solution to Problem \ref{original_problem_contract_signal}, which can be found in time polynomial in the instance size and $\log(1/\zeta)$.
\end{theorem}
\begin{proof} By Lemma \ref{lemma_recommenddirect}, without loss of generality, we  can  let the signal space be $\Sigma = \langle n \rangle$ and consider a direct and IC mechanism. Similar to \cite{castiglioni2023selling,gan2022optimal,castiglioni2024reduction}, we apply relaxation to Problem \ref{original_problem_contract_signal} by replacing the terms $\pi(s|\theta) p^{s, \theta}$ with a new variable $z^{s, \theta}$. Hence, Program \ref{original_problem_contract_signal} can be reformulated as:
\begin{equation}\label{relaxted_programzeqpip}
\begin{aligned}
    \max_{\pi, z} \quad & \sum_{\theta} \mu(\theta) \sum_{s \in \langle n \rangle}  \langle F^{\theta}_{s}, \pi(s|\theta)r - z^{s, \theta}\rangle \\
    \textnormal{subject.to.} \quad & \sum_{\theta} \mu(\theta)  \Big[ \langle F^{\theta}_{s},  z^{s, \theta} \rangle - \pi(s|\theta) c_{s} \Big] \ge \sum_{\theta} \mu(\theta)  \Big[ \langle F^{\theta}_{i}, z^{s, \theta} \rangle - \pi(s|\theta)c_{i} \Big], \quad \forall i, s \in \langle n \rangle \\
& \sum_s \pi(s|\theta) =1, \quad \forall \theta \in \Theta\\
& z^{s, \theta} \ge 0, \quad \forall s \in \langle n \rangle, \forall  \theta \in \Theta.
\end{aligned}
\end{equation}
Notice that Program (\ref{relaxted_programzeqpip}) is linear and hence admits a maximum. Moreover, it can be solved efficiently. After solving Program (\ref{relaxted_programzeqpip}), we want to recover a solution to the original problem.
If $\pi(s|\theta)>0$, it is straightforward to recover $p^{s,\theta}$ from $z^{s,\theta}$ as follows:
\begin{equation}\label{recovingincontract}
p^{s, \theta} = \frac{z^{s, \theta}}{\pi(s|\theta)} \in \mathbb{R}^{m}_+.
\end{equation} 
However, there might exist some \emph{irregular} solutions where $z^{s, \theta} \neq 0$ but $\pi(s|\theta)=0$. Lemma \ref{lemma_epsic} shows a simple method that converts a solution with \emph{irregular} pairs $(z^{s, \theta}, \pi(s|\theta))$ to a \emph{regular} one, i.e., without irregular pairs, through retaining a $\xi$-IC solution with a small utility loss to the principal, where $\xi$ can be arbitrary small.

\begin{lemma}\label{lemma_epsic}
    For any $\xi>0$, given a solution to Problem (\ref{relaxted_programzeqpip}), there exists a polynomial-time procedure that returns a $\xi$-IC solution that decreases by at most $\xi>0$ the principal's utility and such that there are no {\it irregular} $(z^{s, \theta}, \pi(s|\theta))$ pairs. 
\end{lemma}

The proof of Lemma \ref{lemma_epsic} is by the observation that slightly perturbing the signaling scheme does not change the agent's best-response action if slightly relaxing the IC constraints. Hence, by slightly perturbing the signaling scheme, we can eliminate the irregular pairs. We achieve this by the reserving and redistributing technique, which first reserves a small portion of probability with losing a small principal utility and then redistributes the reserved probability to irregular pairs so that those pairs will have $\pi(s|\theta)>0$.

Hence, applying Eq.~\eqref{recovingincontract} to the regular solution that we recovered through Lemma~\ref{lemma_epsic}, we obtain a $\xi$ optimal and $\xi$-IC solution for the Problem \ref{original_problem_contract_signal}. The next lemma shows that any approximately IC solution to the Problem \ref{original_problem_contract_signal} can be converted to an IC mechanism by losing some small principal utility. This is achieved by slightly perturbing the contracts with linear contracts similarly to \cite{dutting2019simple,castiglioni2022designing}

\begin{lemma}\label{lemma_eps_toicbyloss}
    Given any $\xi$-IC solution to Problem \ref{original_problem_contract_signal}, there exists a polynomial-time procedure that returns a IC solution by losing at most $(n+1)\sqrt{\xi}$ principal utility, where $n$ is the number of  actions.
\end{lemma}

    By Lemma \ref{lemma_epsic} and \ref{lemma_eps_toicbyloss}, the  utility loss to the principal is at most $\xi+(n+1)\sqrt{\xi}$. By taking $\xi = (\frac{\zeta}{2(n+1)})^2$, we guarantee that the loss of principal utility (with respect to the optimal mechanism) is at most $\zeta$. Finally, all the above steps can be done in polynomial time. This concludes the proof for the theorem.
\end{proof}

\section{Menus of Explicit Contracts}\label{sec:explicitycontractdesing}

In this section, we consider the joint design problem with a menu of explicit contacts.  In light of the similarity to the classical Bayesian persuasion problem \cite{kamenica2011bayesian,dughmi2019algorithmic}, it is reasonable to expect that it is sufficient to consider direct mechanisms in this problem as well. However, our first result shows that in this model with less flexible contracts, direct mechanisms are suboptimal. This is in sharp contrast with the positive result in Lemma \ref{lemma_recommenddirect}. Indeed, similar phenomenons were observed in \cite{yang2024computational,gan2023robust} when concerning the robustness of Bayesian persuasion. Below, we formally state our negative result.

\begin{proposition}\label{sumoptimaiontamenuscontr}
    (Suboptimality of direct mechanisms) There exist simple instances (with $2$ actions and $2$ outcomes) such that direct mechanisms fail to achieve the supremum of Program~\ref{original_problem_contract_menu}.
\end{proposition}
\begin{proof}
    We prove the proposition by constructing an series of examples in which direct mechanisms are suboptimal. Let $0<\delta<0.5$.  We consider instances with $2$ outcomes and $2$ actions (i.e., $a_0$ and $a_1$). To ease the exposition, we do not normalize the reward and cost to the interval $[0, 1]$, and we define $r = [10, 0]$ and $c = [0, 5]$. There are three states $\theta_1, \theta_2, \theta_3$, where the prior probabilities are $\mu(\theta_1) = 1-2\delta$ and $\mu(\theta_2) = \mu(\theta_3) = \delta$. The probability matrices mapping actions to outcomes are defined as:
\begin{equation*}
        F^{\theta_1} = \left[
    \begin{array}{cc}
         1& 0 \\ [3pt]
         1& 0 
    \end{array}
    \right] \quad \quad
     F^{\theta_2} = \left[
    \begin{array}{cc}
         0.1& 0.9 \\ [3pt]
         0.8& 0.2 
    \end{array}
    \right] \quad \quad 
      F^{\theta_3} = \left[
    \begin{array}{cc}
         0.2 & 0.8 \\ [3pt]
         1 & 0  
    \end{array}
    \right] 
\end{equation*}
The key idea is to show that sending two signals (i.e., recommending one action for each signal) achieves less utility than applying the full information revelation strategy. Notice that full revelation requires three signals. 

First, we calculate the optimal principal's utility for full revelation. It is not hard to verify the following: 
\begin{itemize}
    \item For state $\theta_1$, the optimal contract is $p^{\theta_1} = [0, 0]$ and the agent's best response is action $a_0$. Hence, the principal's utility is $10$.
    \item For state $\theta_2$, the the optimal contract is $p^{\theta_2} = [\frac{50}{7}, 0]$ and the agent's best response is action $a_1$. The principal's utility is $\frac{16}{7}$.
    \item For state $\theta_3$, the the optimal contract is $p^{\theta_3} = [\frac{25}{4}, 0]$ and the agent's best response is action $a_1$. The principal's utility is $\frac{15}{4}$. 
\end{itemize}
To summarize, the total utility of the principal under full information revelation is $10(1-2\delta) + \frac{169}{28}\delta$.

Next, we show that any other signaling scheme with only two signals (i.e., a signal $s_0$ inducing action $a_0$ and a signal $s_1$ inducing action $a_1$) will achieve lower principal's utility than full revelation. Assume that a signal  $s \in \{s_0, s_1\}$ induces a {\it non-normalized} posterior distribution   over states $q = [q_1, q_2, q_3]$, i.e.,  $q_i = \mu(\theta_i)\pi(s|\theta_i)$ for each $i$. Then, the signal induces an expected probability matrix $\bar{F}$ defined as: 
    \begin{equation*}
    \bar{F} = \frac{1}{Q}\left[
    \begin{array}{cc}
         q_1 + 0.1q_2+ 0.2q_3& 0.9q_2 + 0.8q_3 \\ [3pt]
         q_1 + 0.8q_2 + q_3& 0.2q_2 
    \end{array}
    \right]
\end{equation*}
where $Q = q_1+q_2+q_3$. Our first step is to show that it is not optimal to ‘‘mix'' $\theta_1$ with the other two states $\theta_2$ and $\theta_3$, i.e., any signal with  $q_1>0$ and at least one of $q_2> 0$ and $q_3 > 0$ is not optimal. We divide the analysis into two cases. 

{\bf Case 1: The best-response action is $a_0$.}  If the agent plays $a_0$, then it is optimal to employ contract $p^{s} = [0, 0]$. Hence, the principal's utility from signal $s$ is $\frac{10q_1 + q_2 + 2q_3}{Q}$.
We show that if we  decompose the signal $s$ into three signals $s_{\theta_1},s_{\theta_2},s_{\theta_3}$ such that each signal $s_\theta$ induces as posterior $\theta$, the principal utility increases.
In particular, similarly to what we showed for the full revelation signaling scheme, the principal can get utility $\frac{10q_1 + \frac{16}{7}q_2 + \frac{15}{4}q_3}{Q} > \frac{10q_1 + q_2 + 2q_3}{Q}$.
This implies that for any signal inducing action $a_0$, mixing $\theta_1$ with $\theta_2$ or $\theta_3$ decreases the utility with respect to full revelation. 

{\bf Case 2: The best-response action is $a_1$.} Let the designed contract be $p = [x, y]$. By the IC constraints for action $a_1$, we have 
\[
x\frac{(q_1 + 0.8q_2 + q_3)}{S} + y \frac{(0.2q_2)}{Q} -5 \ge x \frac{(q_1 + 0.1q_2+ 0.2q_3)}{Q} + y\frac{(0.9q_2 + 0.8q_3)}{Q},
\]
which is equivalent to 
\[
x\frac{(0.7q_2 + 0.8q_3)}{Q} -5 \ge y \frac{(0.7q_2 + 0.8q_3)}{Q}.
\]
Simple calculations show that in the optimal solution $y=0$. Hence, the optimal contract is $p= [\frac{5Q}{0.7q_2 + 0.8q_3}, 0]$. 

Similar to Case 1,  we show that the full revelation scheme gives higher utility. That is to show 
\[
10q_1 + \frac{16}{7}q_2 + \frac{15}{4}q_3\ge
Q[\frac{q_1 + 0.8q_2 + q_3}{Q}(10 - \frac{5Q}{0.7q_2 + 0.8q_3})]
\]
where we remark that we can assume $10 - \frac{5Q}{0.7q_2 + 0.8q_3}>0$ and otherwise, the inequality holds automatically. Sufficiently we only need to show the following holds,
\[
 \frac{16}{7}q_2 + \frac{15}{4}q_3>
[{ 0.8q_2 + q_3}](10 - \frac{5(q_2+q_3)}{0.7q_2 + 0.8q_3})
\]
To show the inequality, it is without loss to assume $q_2+q_3 =1$. Then, it can be shown that the above holds if $1\ge q_2\ge 0$. This implies that for any signal inducing action $a_1$, mixing $\theta_1$ with $\theta_2$ or $\theta_3$ decreases the utility with respect to full revelation.

Finally, we know that the optimal direct mechanism with $2$ signals recommending action $a_0, a_1$ must be
\begin{equation*}
    \pi^{D} = \left[
    \begin{array}{cc}
         1& 0 \\ [3pt]
         0& 1 \\ [3pt]
         0 & 1
    \end{array}
    \right]
\end{equation*}
It is not hard to calculate the optimal contracts are $p^{s_0} = [0, 0]$ and  $p^{s_1} = [\frac{20}{3}, 0]$. The total utility is 
$10(1-2\delta) + 6\delta < 10(1-2\delta) + \frac{169}{28}\delta$. Hence, we prove that full revelation is the optimal strategy, which concludes the proof.
\end{proof}

In Bayesian persuasion, a revelation-principle-style argument is usually applied to join two signals as one if they induce the same best-response action for the agent~\cite{kamenica2011bayesian, dughmi2019algorithmic}.  However, the constructed example shows a counter-intuitive result: it could be suboptimal to join such two signals. 
This observation poses an intriguing question: how many signals are needed to implement an optimal mechanism for Problem~\ref{original_problem_contract_menu}? We leave it as an open question for future works.

Our previous result shows that optimal mechanisms might be complex and include many signals. Our next theorem complements this result showing that approximately optimal mechanism with a menu of explicit contracts cannot be computed efficiently. Indeed,  Problem~\ref{original_problem_contract_menu}  does not admit an additive PTAS unless $\P=\NP$. Interestingly, in our reduction optimal contracts required high payments. This is a necessary condition since, as we show in Section~\ref{sec:bounded} approximately optimal contracts with bounded (constant) payments can be approximated in quasi-polynomial time.

\begin{theorem}\label{menuhandressrprof}
    There is a constant $\rho>0$ such that it is NP-Hard to find a $\rho$-additive approximation of the Problem \ref{original_problem_contract_menu}.
\end{theorem}

\begin{proof}
We prove the hardness by reducing from a problem related to finding disjoint independent sets. In particular, \citet{khot2012hardness} show that given an integer $k$, an integer $q$
such that $q \ge 2^{k}+1$, and an arbitrarily small constant $\epsilon>0$, the following problem is NP-Hard:

Given a graph $(V,E)$, decide whether:
\begin{itemize}
    \item {\bf ``Yes" case:} The are $q$ disjoint independent sets $V_1, \ldots, V_q \subseteq V$, such that $|V_i|=\frac{1 - \epsilon}{q}|V|$ for each $i \in [q]$.
    \item {\bf``No" case:} There are no independent sets of size $\frac{1}{ q^{k+1}}|V|$.
\end{itemize}

We reduce from the previous problem with $k=1$, $q=3$, and $\epsilon=\frac{1}{100}$.
Then, either there exist $3$ disjoint independent sets of size at least $\bar k = \frac{33}{100}|V|$,  or all the independent sets have size at most $\hat k = \frac{1}{9}|V|$.

\textbf{Construction.} For each vertex $v \in V$, there is one state of nature $\theta_v$ and one outcome $\omega_v$. Moreover, there are two additional outcomes $\omega^*$ and $\omega_\varnothing$. The prior $\mu$ is uniform over all the states. The reward on outcome $\omega^*$ is
$r_{\omega^*}=1$, while the reward of every other outcome $\omega \neq \omega^*$ is $r_{\omega}=0$. We define a parameter $\delta = \frac{1}{|V|10^{5}}$ that will be used in the following. 

Next, we elaborate on the construction of the matrix $F^{\theta_v}$ for all $\theta_v \in \Theta$. All other non-specified entries are $0$. Figure \ref{proceduremultisingleprocess} Provides an example of the construction.

\begin{itemize}
    \item There is an action $a^*$, whose cost is  $c_{a^\star}=1/4$. For each $\theta_v \in \Theta$, let $F^{\theta_v}_{a^*, \omega^*}=\frac{1}{2}$. Additionally, let $F^{\theta_v}_{a^*, \omega_v}=\delta$ and $F^{\theta_v}_{a^*, \omega_\varnothing}=\frac{1}{2}-\delta$. 
    \item There exists an action $\hat a_v$ for each $v\in V$, whose cost is $c_{\hat a_v}=1/16$. Thus, in total, there are $|V|$ such actions $\hat a_v$. For each $\theta_v \in \Theta$, let $F^{\theta_v}_{\hat a_{v'}, \omega^*}=1/4$ for all $v' \in V$. Moreover, let $F^{\theta_{v}}_{\hat a_v, \omega_v}=\hat k \delta$ and $F^{\theta_v}_{\hat a_v, \omega_{\varnothing}}=\frac{3}{4}- \hat k \delta$. For $v' \neq v$, let $F^{\theta_v}_{\hat a_{v'}, \omega_{\varnothing}}=\frac{3}{4}$. 
\item There is an action $\bar a_v$ for each $v\in V$,  whose cost is $c(\bar a_v)=0$. For each action $a_v,  v\in V$,   let $F^{\theta_{v'}}_{\bar a_{v}, \omega_{v}} = \frac{\delta}{10}$ and $F^{\theta_{v'}}_{\bar a_{v}, \omega_\varnothing} = 1- \frac{\delta}{10}$ for all $ v'\in V$. 
\item For each $v\in V$, we add an action $\tilde a_v$ whose cost is $c(\tilde a_v)=0$. Consider each state $\theta_v, v\in V$. If there is an edge between nodes $v$ and $v'\neq v$, i.e., $(v, v')\in E$, let $F^{\theta_{v}}_{\tilde a_{v'}, \omega_{v'}}=1$. If $(v, v')\notin E$, let $F^{\theta_{v}}_{\tilde a_{v'}, \omega_{\varnothing}}=1$. Finally, $F^{\theta_{v}}_{\tilde a_{v}, \omega_{\varnothing}}=1$.
\end{itemize}

\begin{figure}[t]
\centering
\includegraphics[scale=0.5]{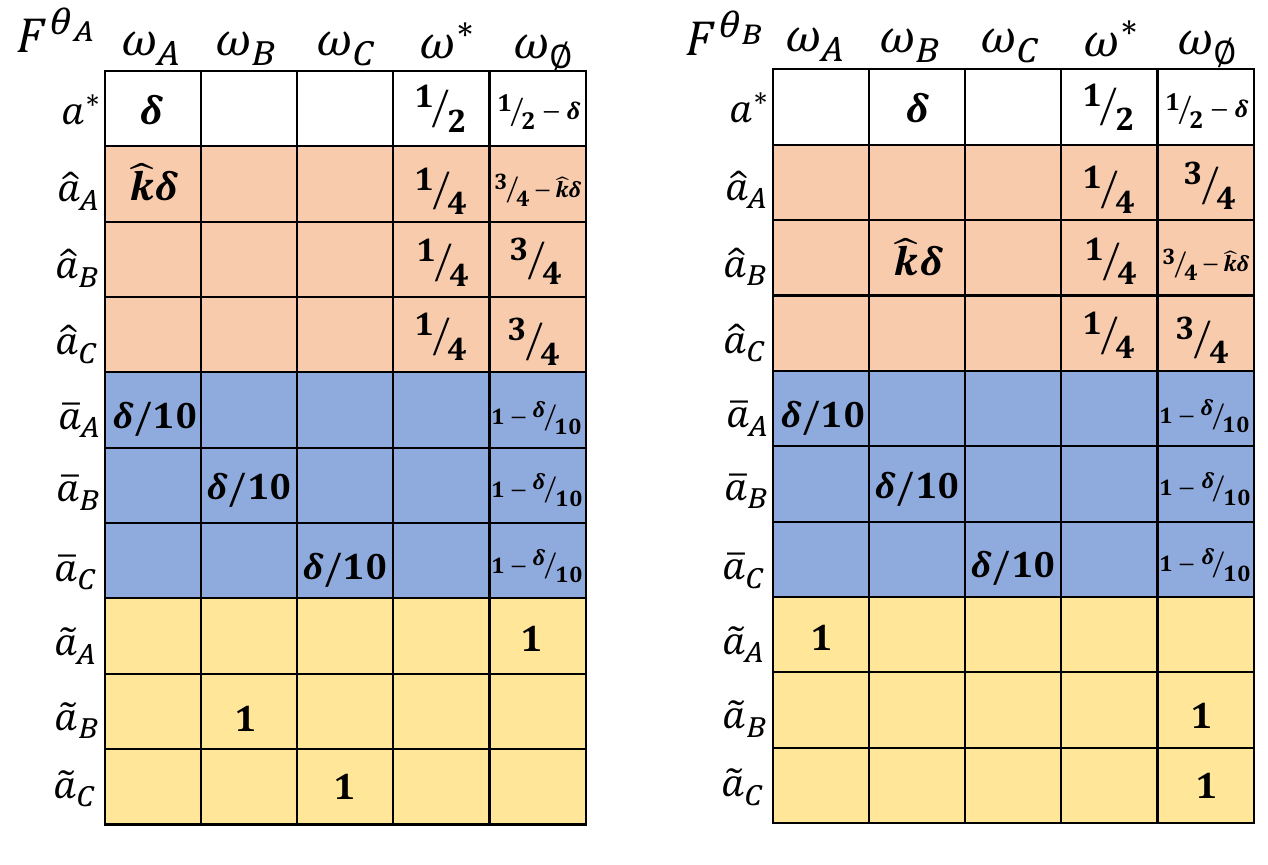}
\caption{An example for the reduction of Theorem \ref{menuhandressrprof}, where the graph consists of $3$ nodes A, B, C such that A is connected to B and C and B and C are not connected. All the empty entries are zeros. }
\label{proceduremultisingleprocess}
\end{figure}

Intuitively, our instances guarantees two properties. First, the principal must incentivize the agent to choose action $a^*$ in order to achieve a large utility. Second, incentivize action $a^*$ (with a small payment) is related to find a large independent set.
In particular, the introduction of action $\hat{a}$ ensures that the payments are ‘‘sufficiently uniform'' over the outcomes associated with nodes in $V$. 
Moreover, actions $\bar{a}$ and $\tilde{a}$  excludes the situations where a signal and the associated contracts would simultaneously assign high probabilities and large payments to two nodes connected by an edge, i.e., not in an independent set.

By this delicate construction, we are able to build a connection between the principal's utility and the size of an independent set.
Our proof employs the  following two key lemmas: Lemma \ref{lm:yes} and \ref{ourlemmanocase}.
In Lemma \ref{lm:yes} we exploit that the independent sets are almost a partition of $V$. In this way, our signaling scheme can exploit three signals where each signal induces an uniform posterior over one of the independent set, guaranteeing an large principal's utility (almost) independently from the realized state of nature. This would not be possible reducing from independent set.
The main challenge comes from the proof of Lemma \ref{ourlemmanocase}. To prove that, we first show that it is without loss of generality to consider contracts with zero payments on the outcome $\omega_{\varnothing}$. Then, we prove by contradiction that if there exists a signal inducing principal's utility larger than $\frac{31}{160}$ then there exists a large independent set.

\begin{lemma}\label{lm:yes}
(Completeness) If ``Yes" case holds, then there exists a mechanism such that the utility of  principal is at least $\frac{13563}{63040}$.
\end{lemma}

\begin{lemma}\label{ourlemmanocase}
   (Soundness) If ``No" case holds, the optimal mechanism achieves an expected principal's utility strictly smaller than $\frac{31}{160}$.
\end{lemma}

The previous two lemmas show that the gap between the utilities of these two cases is $\rho = \frac{13563}{63040} - \frac{31}{160} \ge 0.021$. This proves that it is NP-Hard to approximate our problem within a $\rho$ factor, concluding the proof.
\end{proof}

\subsection{Proof of Lemma~\ref{lm:yes}}

Recall that we have $q=3$ independent sets.    Consider the following mechanism. The mechanism sends a signal $s_i$ for eac independent set $V_i$ with $i=1, 2, 3$, and an additional signal $s_\varnothing$.
For each independent set $V_i$, we let $\pi(s_i|\theta_v)=1$ for each $v \in V_i$.
Moreover, for any $v\notin V_1 \cup V_2 \cup V_3$, we let $\pi(s_\varnothing|\theta_v)=1$.
All the other probabilities are set to $0$. It is easy to see that the signaling scheme is well-defined.

Next, we define the payments.
For each signal $s_i$, we set $p^{s_i}_{\omega_{v}}=\frac{3}{16}\frac{\bar k}{\delta (\bar k-\hat k)}$ for outcome $\omega_{v}$ with $v \in V_i$, and payment $0$ on other outcomes.
For signal $s_\varnothing$ we set the payment to $0$ on every outcome.

Consider the signal $s_i$ for $i \in \{1,2,3\}$. By the ``Yes" case, we know that the probability of sending signal $s_i$ is at least $\prob{s_i} = \frac{1-\eps}{q} \ge \frac{\bar k}{|V|} = \frac{33}{100}$.
Then, we show that $a^*$ is an agent's best response to signal $s_i$.
Indeed, the agent's expected utility from playing action $a^*$ is
\begin{align*}
    \sum_{\theta} \prob{\theta|s_i} [\langle F^{\theta}_{a^*}, p^{s_i} \rangle -c_{a^*}] &= \sum_{v \in V_i} \delta \prob{\theta_v|s_i}p^{s_i}_{\omega_{v}}-c_{a^*} \\
    & = \bar k \delta \frac{1}{\bar k} \frac{3}{16}\frac{\bar k}{\delta (\bar k-\hat k)} -\frac{1}{4} \tag*{($|V_i| = \bar{k}$ and $\prob{\theta_v|s_i} = \frac{1}{\bar{k}}$)} \\
    & =  \frac{- \bar k+ 4\hat k}{16(\bar k-\hat k)}=\frac{103}{3152}. 
\end{align*}
For $v \in V_i$, the utility of action $\hat a_v$ is at most
\begin{align*}
    \sum_{\theta} \prob{\theta|s_i} [\langle F^{\theta}_{\hat{a}_v}, p^{s_i} \rangle -c_{\hat{a}_v}] &=  \hat{k}\delta \prob{\theta_v|s_i}p^{s_i}_{\omega_{v}}-c_{\hat{a}_v} \\
    & = \hat k \delta \frac{1}{\bar k}\frac{3}{16} \frac{\bar k}{\delta (\bar k-\hat k)}-\frac{1}{16}\\
    &= \frac{-\bar k+4 \hat k}{16(\bar k-\hat k)}=\frac{103}{3152}.
\end{align*}
Since $a^*$ guarantees a larger principal's utility, due to tie-breaking in favor of the principal, the agent will  choose action $a^*$ over $\hat a_v$.

Moreover, for other $v \notin V_i$, the expected payment to the agent is zero. Hence, the agent will not choose action $\hat a_v$.

For $v \in V_i$, the utility of action $\bar a_v$ is at most
\begin{align*}
    \sum_{\theta} \prob{\theta|s_i} [\langle F^{\theta}_{\bar{a}_v}, p^{s_i} \rangle -c_{\bar{a}_v}] = \frac{\delta}{10} p^{s_i}_{\omega_{v}}-c_{\bar{a}_v} 
     = \frac{\delta}{10} \frac{3}{16}\frac{\bar k}{\delta (\bar k-\hat k)}
    =\frac{3}{160} \frac{297}{197}<\frac{103}{3152}.
\end{align*}
Moreover, since the payment on outcome $p^{s_i}_{\omega_{v}} = 0$ for other $v \notin V_i$, the expected payment to the agent is zero. Hence, the agent will not choose action $\bar a_v$.

Finally, for any $v \in  V$, the utility of action $\tilde a_{v}$ is $0$
since by the construction of $p^{s_i}$ which only has positive payments for independent nodes,  $\prob{\theta_v|s_i} p^{s_i}_{\omega_{v'}}=0$ for each $(v,v') \in E$.

Hence, we conclude that the played action is $a^*$ and the expected principal's utility from signal $s_i$:
\begin{align*}
    \sum_{\theta_v} \prob{\theta_v|s_i} \langle F^{\theta_v}_{a^*}, r-p^{s_i}  \rangle &= \sum_{\theta_v} \prob{\theta_v|s_i} F^{\theta_v}_{a^*, \omega^*} r_{\omega^*} -  \sum_{\theta_v: v\in V_i} \prob{\theta_v|s_i} F^{\theta_v}_{a^*, \omega_v} p^{s_i}_{\omega_v}\\
    & = \frac{1}{2} - \frac{1}{\bar{k}} \bar{k} \delta \frac{3}{16}\frac{\bar k}{\delta (\bar k-\hat k)} = \frac{1}{2}- \frac{891}{3152}= \frac{685}{3152}.
\end{align*}

Finally, for signal $s_\varnothing$, the agent will take one action from $\{\bar{a}_v\}_{v\in V}$ and $\{\tilde{a}_v\}_{v\in V}$, which leads to $0$ principal's utility.
Overall, the principal's utility is at least 
\[
\sum_{i\in \{1, 2, 3\}} \prob{s_i} \sum_{\theta_v} \prob{\theta_v|s_i} \langle F^{\theta_v}_{a^*}, r-p^{s_i}  \rangle \ge \frac{685}{3152} \frac{99}{100} =\frac{13563}{63040}. 
\]

\subsection{Proof of Lemma~\ref{ourlemmanocase}}

Recall that the ``No" case implies all the independent sets have sizes strictly smaller than $\hat k$.
For any signal $s$ and its associated induced posterior $(\prob{\theta_v|s})_{\theta_v \in \Theta}$, we want to show that the optimal expected utility of the principal is strictly less than  $\frac{31}{160}$. Then, an averaging argument is sufficient to prove the theorem.

Suppose by contradiction that the utility of principal is at least  $\frac{31}{160}$ under some signal $s$.
First, it is easy to see that to obtain an utility at least $\frac{31}{160}$, the principal must incentivize action $a^\star$.
Indeed, the social welfare of any other action is at upper bounded by that of any action $\hat{a}_v, v\in V$, and in particular by
\[
\sum_{\theta} \prob{\theta|s} [\langle F^{\theta}_{\hat a_v}, p^{s} \rangle -c_{\hat a_v}] = \frac{1}{4}-\frac{1}{16}=\frac{3}{16}.
\]
Let $p$ be the payment scheme induced by signal $s$, where to simplify the exposition we remove the dependency on $s$. We first show that we can set $p_{\omega_{\varnothing}} = 0$. Formally, given the signal $s$, we show that there always exists an an optimal contract with $p_{\omega_{\varnothing}} = 0$.
\begin{claim}
     For any signal $s$, there exists an optimal contract with $p_{\omega_\varnothing} = 0$.
\end{claim}
\begin{proof}
First, since $\delta = \frac{1}{|V|10^{5}} < \frac{1}{4(\hat{k}-1)}$, we have  $\frac{3}{4}-\hat{k}\delta \ge \frac{1}{2}-\delta$. Given any optimal contract $p^*$ for signal $s$, we construct a new contract $p$ such that $p_{\omega_v} = p^*_{\omega_v}$, $p_{\omega_\varnothing} = 0$ and $p_{\omega^*} = p^*_{\omega^*} + 2(\frac{1}{2}-\delta)p^*_{\omega_\varnothing}$. We first show that the action $a^*$ is still incentivized under this new contract.

Note that under the new contract $p$, the incentives for actions $\bar{a}_v$ and $\tilde{a}_v$ are decreased since the payment $p_{\omega_\varnothing}$ is set to be $0$. Consider action $\hat{a}_v$,  we have 
\begin{align*}
\sum_{\theta} \prob{\theta|s} &[\langle F^{\theta}_{a^*}, p \rangle -c_{a^*}]\\&=\frac{1}{2} p_{\omega^*} + \sum_{v' \in V} \delta \prob{\theta_{v'}|s}p_{\omega_{v'}}  - \frac{1}{4}\\
    &=\frac{1}{2}( p^*_{\omega^*} + 2 (\frac{1}{2}-\delta)p^*_{\omega_{\varnothing}})+ \delta \sum_{v' \in V}\prob{\theta_{v'}|s}p^*_{\omega_{v'}}  - \frac{1}{4} \\
&\ge \frac{1}{4} p^*_{\omega^*} +\delta \hat k \prob{\theta_{v}|s}p^*_{\omega_{v}} + p^*_{\omega_{\varnothing}} \sum_{v' \in V} \prob{\theta_{v'}|s}F^{\theta_{v'}}_{a_{v}}(\omega_{\varnothing})  -\frac{1}{16} \tag*{(IC constraint for $a^*$)} \\
& \ge \frac{1}{4} (p_{\omega^*}^* +2 (\frac{1}{2}-\delta)p^*_{\omega_{\varnothing}})+\delta \hat k \prob{\theta_{v}|s}p^*_{\omega_{v}}  -\frac{1}{16}\\
& = \frac{1}{4} p_{\omega^*} +\delta \hat k \prob{\theta_{v}|s}p_{\omega_{v}}  -\frac{1}{16} = \sum_{\theta} \prob{\theta|s} [\langle F^{\theta}_{\hat{a}_v}, p \rangle -c_{\hat{a}_v}],
\end{align*}
where the second inequality follows by $\frac{1}{2}(\frac{1}{2}-\delta) \le \sum_{v' \in V} \prob{\theta_{v'}|s}F^{\theta_{v'}}_{a_{v}}(\omega_{\emptyset})$, which holds since $\frac{3}{4}-\hat{k}\delta \ge \frac{1}{2}-\delta$. This implies that $a^*$ is the agent's best-response action.
Finally, we can verify that the expected payment for action $a^*$ does not change. This concludes the proof.
\end{proof}

Next, we provide some upper bounds on the payment over outcomes $\omega_v$, $v\in V$.

\begin{claim}\label{claim:hard}
    For all $v\in V$, it holds:
   \[\prob{\theta_v|s} p_{\omega_v}\le \frac{1}{2\hat k}  \sum_{v' \in V} \prob{\theta_{v'}|s}p_{\omega_{v'}} - \frac{11}{320\delta \hat k}.\] 
\end{claim}
\begin{proof}
Recall that $a^*$ is the best-response action. Consider any action $\hat{a}_v$, $v\in V$. By the IC constraints, we have  
\[ \frac{1}{2} p_{\omega^*} + \delta \sum_{v' \in V}\prob{\theta_{v'}|s}p_{\omega_{v'}}- \frac{1}{4} \ge \frac{1}{4} p_{\omega^*} +\delta \hat k \prob{\theta_{v}|s}p_{\omega_{v}} -\frac{1}{16},  \]
implying
\begin{align}\label{eq:claim1}
\delta \sum_{v' \in V}\prob{\theta_{v'}|s}p_{\omega_{v'}}  \ge  \frac{3}{16} +\delta \hat k \prob{\theta_{v}|s}p_{\omega_{v}}- \frac{1}{4}p_{\omega^*}. 
\end{align}
Since the principal's utility is at least $\frac{31}{160}$ by assumption, it must be the case that
\[ \frac{1}{2}- \frac{1}{2}p_{\omega^*} - \delta \sum_{v' \in V}\prob{\theta_{v'}|s}p_{\omega_{v'}}   \ge \frac{31}{160}\]
implying
\begin{align}\label{eq:claim2}  -\frac{1}{2}\delta \sum_{v' \in V}\prob{\theta_{v'}|s}p_{\omega_{v'}}  \ge -\frac{49}{320}+\frac{1}{4}p_{\omega^*}.
\end{align}
Then, summing Eq.~\eqref{eq:claim1} and Eq.~\eqref{eq:claim2}, we get:
\[ \frac{1}{2} \delta \sum_{v' \in V}\prob{\theta_{v'}|s}p_{\omega_{v'}}  \ge \frac{11}{320} +\delta \hat k \prob{\theta_{v}|s}p_{\omega_{v}} ,  \]
and
\[ \prob{\theta_{v}|s}p_{\omega_{v}} \le \frac{1}{2 \hat k}  \sum_{v' \in V}\prob{\theta_{v'}|s}p_{\omega_{v'}}-\frac{11}{320\delta \hat k}. \]
This concludes the proof.
\end{proof}

We proceed providing some upperbounds on the payment on each outcome $\omega_v$, and on the product between the payment on outcome $\omega_v$ and the posterior probability of a ‘‘connected'' state $\omega_{v'}$.

\begin{claim}\label{claim:small}
    It holds:
    \[ p_{\omega_v} \le \frac{10}{\delta} \quad \forall v \in V\]
    and
    \[ \prob{\theta_{v'}|s} p_{\omega_v}\le 1 \quad \forall (v',v)\in E. \]
\end{claim}
\begin{proof}
Under signal $s$ and contract $p$, let $\E(a^*, p)$ denote the expected payment of playing action $a^*$. First, due to the principal's utility is at least $\frac{1}{2}- \E(a^*, p)\ge \frac{31}{160}$,  we have $\E(a^*, p) \le 1$. Consider any $\bar{a}_v$ with $v\in V$. By the IC constraint, we have that 
\[\sum_{\theta} \prob{\theta|s} [\langle F^{\theta}_{\bar{a}_v}, p \rangle -c_{\bar{a}_v}] = \frac{\delta}{10} p_{\omega_v}\le \E(a^*, p)-\frac{1}{4}\le 1 , \]
implying the first part of the lemma.
Similarly, consider any $\tilde a_v$, $v \in V$. By the IC constraint, it holds
\[\sum_{\theta} \prob{\theta|s} [\langle F^{\theta}_{\tilde{a}_v}, p \rangle -c_{\tilde{a}_v}] = \sum_{v':(v',v)\in E}\prob{\theta_{v'}|s} p_{\omega_v}\le  E(a^*)-\frac{1}{4}\le 1\]
implying the second part of the claim.
\end{proof}

Then, we define the set of nodes which state has a large product between posterior probability and payment. Formally, we define the set of nodes $\hat V=\{v \in V: \prob{\theta_v|s} p_{\omega_{v}}  \ge \rho\}$, where $\rho=  \frac{1}{4|V|}\sum_{v \in V} \prob{\theta_{v}|s}p_{\omega_{v}}$. 
Intuitively, $\hat V$ includes all the nodes (i.e., outcomes) that provide a sufficiently large portion of the agent's utility when the agent plays $a^*$.
We show that $\hat V$ must includes at least $ \frac{3}{2}\hat k$ nodes.
\begin{claim}\label{claim:large}
    The size of set $\hat V$ is $|\hat V|\ge \frac{3}{2}\hat k $.
\end{claim}
\begin{proof}
First, notice that 
\begin{align*}
\sum_{v \in V} \prob{\theta_{v}|s}p_{\omega_{v}}
&=\sum_{v \in \hat V} \prob{\theta_{v}|s}p_{\omega_{v}} + \sum_{v \notin \hat V} \prob{\theta_{v}|s}p_{\omega_{v}}\\
&\le |\hat V| \frac{1}{2\hat k} \sum_{v \in V} \prob{\theta_{v}|s}p_{\omega_{v}} + \sum_{v \notin \hat V} \prob{\theta_{v}|s}p_{\omega_{v}},  \\
&\le |\hat V| \frac{1}{2\hat k} \sum_{v \in V} \prob{\theta_{v}|s}p_{\omega_{v}} + |V| \frac{1}{4|V|}\sum_{v \in V} \prob{\theta_{v}|s}p_{\omega_{v}}\\
&= |\hat V| \frac{1}{2\hat k} \sum_{v \in V} \prob{\theta_{v}|s}p_{\omega_{v}} + \frac{1}{4}\sum_{v \in V} \prob{\theta_{v}|s}p_{\omega_{v}},
\end{align*}
where the first inequality comes from Claim~\ref{claim:hard} and the second from the definition of $\hat V$. This directly implies $|\hat V|\ge \frac{3}{2}\hat k$.
\end{proof}

In the following, we prove that $\hat{V}$ is an independent set, reaching a contradiction. As the first step,  we show that for $v \in \hat V$, both the posterior probability of state $\theta_v$ and the payment on outcomes $\omega_v$ are sufficiently large.

\begin{claim}\label{claim:highProb}
    For each $v \in \hat V$, it holds:
    \[p_{\omega_v}\ge \frac{1}{|V|} \frac{11}{640\delta}  \]
    and 
    \[ \prob{\theta_v|s}\ge \frac{1}{|V|} \frac{11}{6400}.  \]
\end{claim}

\begin{proof}
We start by observing that  Claim~\ref{claim:hard} implies
\[ \sum_{v' \in V} \prob{\theta_{v'}|s}p_{\omega_{v'}}\ge \frac{11}{160\delta}. \]
Hence, by definition of $\hat V$ it must be the case that
\[\prob{\theta_v|s} p_{\omega_{v}} \ge \frac{1}{4|V|} \frac{11}{160\delta}\]
for each $v \in \hat V$. This directly implies the first part of the claim.
Moreover, by Claim~\ref{claim:small},  
$p_{\omega_v}\le\frac{10}{\delta}$ holds, 
which implies
\[\prob{\theta_v|s}\ge \frac{1}{|V|} \frac{11}{6400}.\]

\end{proof}

Equipped with the above technical results, we are ready to show that $\hat V$ is an independent set.
 \begin{claim}
    The set $\hat V$ is an independent set, i.e., there are no edges between any two nodes $v,v'\in \hat V$ such that $(v,v')\in E$
\end{claim}
\begin{proof}
    Suppose by contradiction that there exists two nodes $v,v'\in \hat V$ such that $(v',v)\in E$.
    By Claim~\ref{claim:highProb} and the second part of Claim~\ref{claim:small}, we have 
    \[ 1\ge \prob{\theta_{v'}|s} p_{\omega_v}\ge \frac{1}{|V|}\frac{11}{6400} \frac{11}{640\delta}>1, \]
    where the last inequality holds due to $\delta = \frac{1}{|V|10^{5}}$. Hence, we reach a contradiction, concluding that $\hat V$ is an independent set.
\end{proof}

Finally,  since $\hat V$ is an independent set, we should have that $|\hat V|\le \frac{1}{9}|V|$ by the assumption  of ``No" case. This is in contradiction with Claim~\ref{claim:large}, concluding the proof of Lemma \ref{ourlemmanocase}. 

\subsection{Restriction to Bounded Payments}\label{sec:bounded}
Interestingly, in the hard instances that we designed, optimal contracts require high payments. Here, we show that almost optimal \emph{bounded} menus of contracts can be computed in quasi-polynomial time.
Formally, we assume that all payments are bounded by some constant $B>0$. Note that the bounded payment is a common assumption considered in many recent works on contract design \cite{chen2024bounded, bacchiocchi2023learning, zhu2023sample}.

\begin{proposition}\label{prop:K}
    For any $\eps>0$ and $B>0$, there exists an algorithm running in $\poly(n^{\frac{\log(2nB/\eps)}{2B^2\eps^4}}, |\Theta|)$ time, which returns an incentive-compatible solution with principal's utility at least $\opt- 4\eps$.
\end{proposition}

\section{Single Explicit Contract}\label{sec:single}

\begin{figure}[t]
\centering
\includegraphics[scale=0.56]{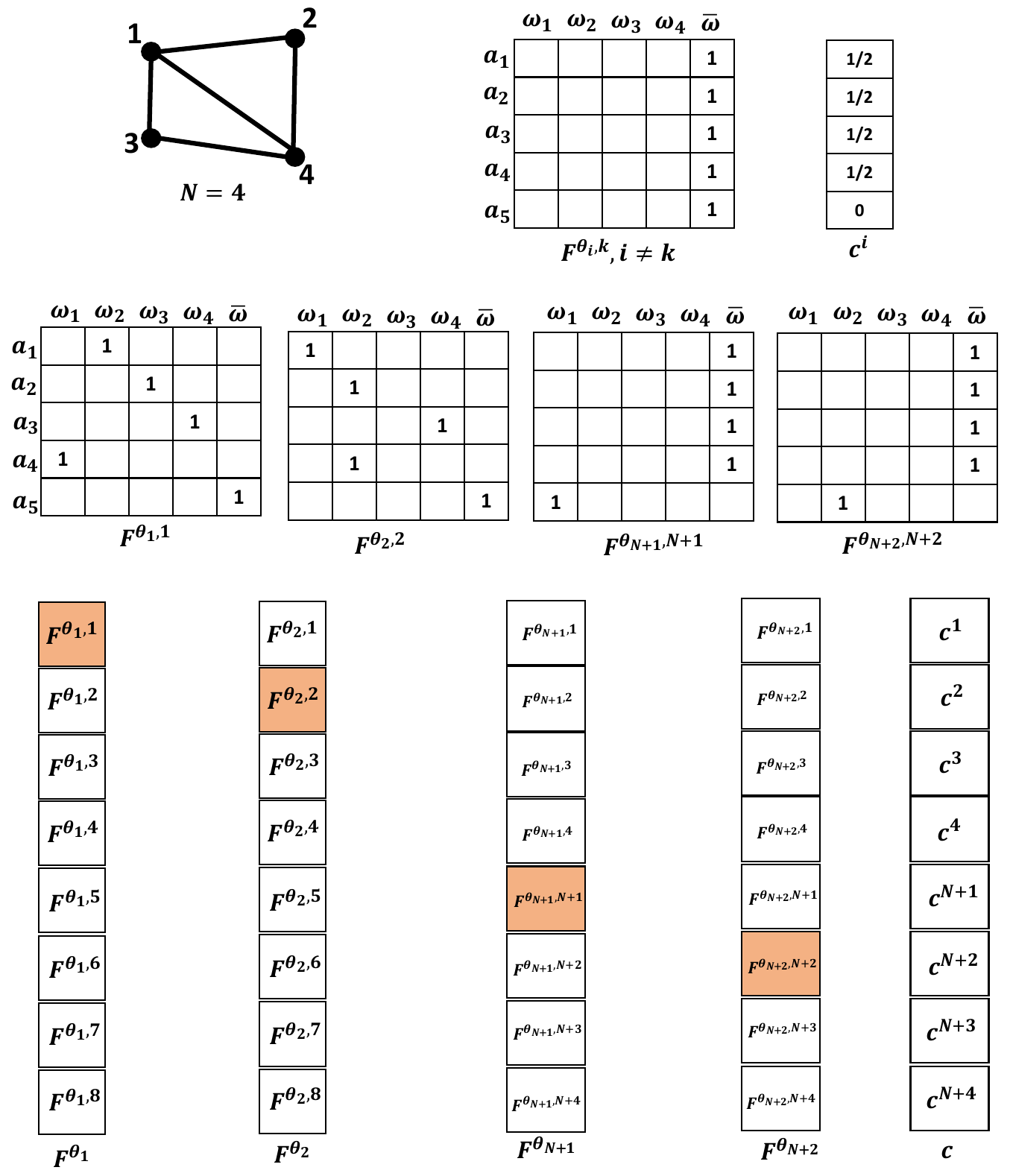}
\caption{An example for the reduction of Theorem \ref{singlcontractishard}, where the empty entries are all zeros.}
\label{processsreductionsinglecontacthard}
\end{figure}

In this section, we focus on the computation of a mechanism with a single explicit contract, which has the least flexibility. As the flexibility of design further decreases, one may expect that a direct mechanism is also suboptimal as Proposition \ref{sumoptimaiontamenuscontr}. However, in sharp contrast to the settings with an explicit menu and ambiguous contracts, a direct mechanism exists and the optimum is indeed achievable. 
Intuitively, this follows since all signals share the same contract. Hence, once the optimal contract is fixed, the information design problem is linear and a revelation-principle-like argument holds.

\begin{proposition}\label{singlcontractdirectmech}
    Problem \ref{original_problem_contract_single_exp} always admits an  optimal direct and IC mechanism.
\end{proposition}

Although the case with a single explicit contract has a simpler information structure than that with an explicit menu, we show that they still cannot be computed efficiently.
In particular, we show that the problem does not admit a PTAS unless $\P=\NP$.
Our main idea is to reduce our problem to the Bayesian contract design problem. Intuitively, this can be achieved by ‘‘forcing'' the signaling scheme to be full revelation.
It is easy to see that in this way the problem of designing a contract $p$ is exactly a Bayesian contract design problem.
Clearly, it is not possible to explicitly constrain the signaling scheme to be full revelation. Indeed, the main challenge of the proof is to design hard instances in which the optimal signaling scheme is directly revealing the state. On top of this, we need to design hard instances of the Bayesian principal-agent problem.
Our reduction does not provide a black-box reduction from the principal-agent problem but requires instances with specific properties.
We design these hard instances of the principal-agent problem taking inspiration from the reduction of~\cite{guruganesh2021contracts}.

\begin{theorem}\label{singlcontractishard}
    There exists a constant $\eps > 0$ such that it is NP-hard to approximate Problem~\ref{original_problem_contract_single_exp} within a multiplicative factor $(1-\eps)$. 
\end{theorem}
\begin{proof}
    We reduce from the problem of finding the smallest dominating set on a graph $G=(V, E)$ where each vertex has a maximum degree at most $3$ \cite{chlebik2008approximation}. A {\it dominating set} for a graph $G$ is a subset $D$ of its vertices, such that any vertex of $G$ is in $D$ or has a neighbor in $D$.

    \medskip

    \textbf{Construction.}  Given an instance of  dominating set on a graph $G=(V, E)$ with $N$ vertices and maximum degree $3$, we construct our instance as follows. 
    There are $2N$ states $\Theta = \{\theta_1, \theta_2, \cdots \theta_{2N}\}$. Each state $\theta_i$ is associated with one action-outcome probability matrix $F^{\theta_i} \in \mathbb{R}^{(2N \times 5) \times (N+1)}$, i.e., each matrix $F^{\theta_i}$ has $10N$ actions and $N+1$ outcomes.
    $F^{\theta_i}$ can be interpeted as a concatenation of $2N$ submatrices $F^{\theta_i, k} \in \mathbb{R}^{5\times (N+1)}$, i.e., $F^{\theta_i} = [F^{\theta_i, 1}; F^{\theta_i, 2}; \dots; F^{\theta_i, 2N}]^\top$. We use $A^i=\{a_{i,1},\ldots, a_{i,5}\}$ to denote the set of actions in the $i^{\textnormal{th}}$ submatrix. We use $\bar{\omega}$ to denote the dummy outcome. Each vertex in graph $G$ is associated with one outcome $\omega$ in matrix $F^{\theta_i}$.

    Next, we construct the matrix $F^{\theta_i}$. {Figure \ref{processsreductionsinglecontacthard}} depicts an example of the construction.
    For $i \in \{1, 2, \dots, N\}$, the $k^{\textnormal{th}}$ submatrix $F^{\theta_i, k}$ is constructed as follows, where with a slight abuse of notation, we define $F^{\theta_i,k}_{a_{t,\omega}}=F^{\theta_i,k}_{a_{k,t},\omega}$ 
    \begin{itemize}
        \item If $k\neq i$, then for each $t \in \{1, 2, \dots, 5\}$, let $F^{\theta_i, k}_{a_t, \bar{\omega}} = 1$ and  $F^{\theta_i, k}_{a_t, \omega} = 0$ for $\omega \neq \bar{\omega}$.
        \item If $k = i$, we first order the neighbors of vertex $i$ in a non-decreasing order. 
        Denote $\delta_j(i)$ as  vertex $i$'s  $j^{\textnormal{th}}$ neighbor. Then, let $F^{\theta_i, i}_{a_5, \bar{\omega}} =1$, $F^{\theta_i, i}_{a_4, \omega_i} =1$. For $j\in \{1, 2, 3\}$, let $F^{\theta_i, i}_{a_j, \omega_{\delta_j(i)}} =1$. All other entries are set as $0$. If vertex $i$ has less than $3$ neighbors, let vertex $i$ be the dummy neighbor of itself and construct the matrix $F^{\theta_i, i}$ similarly.
    \end{itemize}
Next, we construct the $k^{\textnormal{th}}$ submatrix $F^{\theta_i, k}$ for  $i\in \{ N+1, N+2, \dots, 2N-1, 2N\}$ as follows:
    \begin{itemize}
        \item If $k\neq i$, then for any action $t \in \{1, 2, \dots, 5\}$,  let $F^{\theta_i, k}_{a_t, \bar{\omega}} = 1$  while all other entries are $0$.
        \item If $k = i$, then for $t\in \{1, \dots, 4\}$, let $F^{\theta_i, i}_{a_t, \bar{\omega}} = 1$  and $F^{\theta_i, i}_{a_5, \omega_i} = 1$, while all other entries are $0$.
    \end{itemize}

   The cost vector $c \in \mathbb{R}^{10N}$ consists of $2N$ subvectors $c = [c^1, c^2, \dots, c^{2N}]$ where each $c^{i} = [\frac{1}{2}, \frac{1}{2}, \frac{1}{2}, \frac{1}{2}, 0]$.
    By the above construction, we know that $a_{i,5} \in A^i$ serves as an opt-out action.
   The reward vector $r$ has $r_{\bar{\omega}} = 0$ on the dummy outcome $\bar{\omega}$ while $r_{\omega} = 1$ for all other outcomes $\omega \neq \bar{\omega}$. The uniform prior probability for each state is $\mu({\theta_i}) = \frac{1}{2N}$.

\medskip
\textbf{Utility of full-information  revelation $I$.}  Following a similar analysis as \cite{guruganesh2021contracts}, it is possible to show that full information revelation (denoted as $I$) coupled with a suitable contract defined in the following will give principal's utility $\frac{1}{2} + \frac{1}{4} - \frac{1}{4N} |S|$ where $S$ is the smallest dominating set of graph $G$.
{Indeed, it is not hard to see that under full revelation our problem is equivalent to a Bayesian principal-agent problem in which each type $\theta_i$ has five actions defined according to matrix $F^{\theta_i,i}$, and the distribution over types is uniform.}
The optimal contract is $p^* \in \{\frac{1}{2}, 0\}^{N+1}$ which has payment $p_{\omega_i} =\frac{1}{2}$ with $i\in S$ and $p_{\omega_i} = 0$ with $i\notin S$.

\medskip
\textbf{Full information revelation $I$ is the optimal scheme.} Next, we show that for the instance constructed above, the optimal signaling scheme is full revelation.  

We start Defining the following relaxation of Problem~\ref{original_problem_contract_single_exp}:
\begin{equation}\label{relexted_signle_contract}
\begin{aligned}
    \sup_{\pi, p, a} \quad & U_R (p, \pi, a) \triangleq \sum_s \prob{s} \sum_{\theta} \prob{\theta|s} \langle F^{\theta}_{a(s)}, r - p \rangle \\
    \textnormal{s.t.} \quad 
    & {\sum_{\theta} \mu(\theta)\pi(s|\theta) [\langle F^{\theta}_{a(s)},  p\rangle -c_{a(s)}] \ge \sum_{\theta} \mu(\theta)\pi(s|\theta) [\langle F^{\theta}_{k}, p \rangle -c_{k}]}, \forall a(s)\in A^{i}\setminus \{a_{i, 5}\}, k \in A^{i}, \exists i \\
    & \sum_{\theta} \mu(\theta)\pi(s|\theta) [\langle F^{\theta}_{a(s)}, p\rangle -c_{a(s)}] \ge 0, \quad \forall a(s)= a_{i,5}, \exists i\\
& p \ge 0, \quad \forall s \in \Sigma 
\end{aligned}
\end{equation}
where $\prob{s}$, $\prob{\theta|s}$ are defined as (\ref{posteriorformaitnlleion}) depending on $\pi$, and 
\begin{itemize}
    \item The first constraint ensures that under signal $s$, if the agent applies an action $a(s) \in A^i\setminus \{a_{i,5}\}$ for some $i$, the utility from $a(s)$ only needs to be greater than that of all other actions in $A^i$;
    \item The second constraint ensures that if the agent applies action $a(s) = a_{i,5} \in A^{i}$ for some $i$, the utility only needs to be nonnegative.
\end{itemize}
It is important to note that Problem (\ref{relexted_signle_contract}) is optimized over the action $a(s)$ recommended by the principal. In other words, the principal has the power to ``impose" an action on the agent as long as the constraints are satisfied. As a comparison, the best-response requirement in \ref{original_problem_contract_single_exp} is more restrictive. Although we use $\sup$ instead of $\max$ in Problem (\ref{relexted_signle_contract}), our final results show that it admits a maximum of the same utility as the original \single.


{\bf High-level idea.}  The remaining part of the proof is mainly to show that the optimal solution of Problem (\ref{relexted_signle_contract}) is indeed the optimal solution for \ref{original_problem_contract_single_exp}.  A key observation (Lemma \ref{optimalpbraomegaeq0}) is that it is without loss of generality to consider a contract with $p_{\bar{\omega}} = 0$ on the dummy outcome for Problem (\ref{relexted_signle_contract}), which is due to the relaxation to IC constraints. With this observation, we are able to show in Lemma \ref{fullisbetterlmma} that the full-information revelation $I$ is optimal for (\ref{relexted_signle_contract}). This is proved by showing that given any signaling scheme $\pi$, we can construct a solution with full-information revelation $I$, achieving at least the same utility. Finally, under full information $I$, 
we establish the equivalence between two problems (\ref{relexted_signle_contract}) and \ref{original_problem_contract_single_exp}, which immediately implies that full information revelation $I$ is indeed optimal for \ref{original_problem_contract_single_exp}. Therefore, the hardness results in \cite{guruganesh2021contracts} apply to \ref{original_problem_contract_single_exp}.

\begin{lemma}\label{lemamdirectfor_relsaxed_proble}
    Given a feasible $(p, \pi, a)$ to Problem (\ref{relexted_signle_contract}), there exists a direct scheme with the same $(p, a)$ achieving the same utility.
\end{lemma}
\begin{proof}
    Suppose there exist two signals $s_1$ and $s_2$ such that the principal imposes the same action $\hat{a} = a(s_1)=a(s_2) \in A^i \setminus \{a_{i,5}\}$. We have the following hold 
    \begin{gather*}
        {\sum_{\theta} \mu(\theta)\pi(s_1|\theta) [\langle F^{\theta}_{\hat{a}},  p\rangle -c_{\hat{a}}] \ge \sum_{\theta} \mu(\theta)\pi(s_1|\theta) [\langle F^{\theta}_{k}, p \rangle-c_{k}]}, \forall k \in A^{i} \\
        {\sum_{\theta} \mu(\theta)\pi(s_2|\theta) [\langle F^{\theta}_{\hat{a}},  p \rangle-c_{\hat{a}}] \ge \sum_{\theta} \mu(\theta)\pi(s_2|\theta) [\langle F^{\theta}_{k}, p \rangle-c_{k}]}, \forall k \in A^{i}
    \end{gather*}
By defining a signal $s$ in the new scheme $\pi'$ such that $\pi'(s|\theta) = \pi(s_1|\theta) + \pi(s_2|\theta)$, we have 
\begin{equation*}
     \sum_{\theta} \mu(\theta)[\pi(s_1|\theta) + \pi(s_2|\theta)] [\langle F^{\theta}_{\hat{a}},  p\rangle -c_{\hat{a}}] \ge \sum_{\theta} \mu(\theta)[\pi(s_1|\theta) + \pi(s_2|\theta)] [\langle F^{\theta}_{k},  p \rangle-c_{k}] \quad \forall k \in A^{i}.
\end{equation*}
A similar conclusion holds when action $\hat{a}=a_{i,5}$. Since the utility is linear in $\pi$, we have that the utility of the principal under $\pi'$ does not change. Hence, there exists a direct mechanism with the same $(p, a)$.
\end{proof}

Consider a signaling scheme $\pi$. Conditional on $\pi$, let $(p, a)$ be the tuple of contract and action profile that maximizes Problem (\ref{relexted_signle_contract}). Then, a signal $s$ induces probability mass over states as $x^s = [x^s_1, x^s_2, \dots, x^s_{2N}]$ where $x^s_i = \mu(\theta_i) \pi(s|\theta_i)$. 
Moreover, the posterior over states is $y^s = [y^s_1, y^s_2, \dots, y^s_{2N}]$ {where $y^s_i = \prob{\theta_i|s} =  \frac{x_i^s}{\bar{x}}$ and $\bar {x} = \sum_i x_i^s$.}  
Let $F^s= \sum_{i} y^s_i F^{\theta_i}$ be the average probability matrix with the posterior $y^s$.


\begin{lemma}\label{optimalpbraomegaeq0}
    Conditional on $\pi$, the optimal solution  $(p, a)$ to Problem (\ref{relexted_signle_contract}) has $p_{\bar{\omega}} = 0$.
\end{lemma}
\begin{proof}
    We prove this by contradiction. Suppose  that the optimal solution has $p_{\bar{\omega}} >0$. We construct a new contract $\hat{p}$ such that $\hat{p}_{\bar{\omega}} = 0$ and $\hat{p}_{\omega} = p_{\omega}$ for $\omega \neq \bar{\omega}$. In the following, we show that this new contract $\hat{p}$, together with a compatible action profile $\hat{a}$, is feasible under $\pi$ but gives a larger utility to the principal.

Without loss of generality, consider that under any signal $s$, the principal imposes an action $a(s)$ on the agent.
We first consider the case $a(s) \in A^i$ with $i\in \{1, 2, \dots,  N\}$.
\begin{itemize}
    \item If $a(s) = a_{i,5} \in A^{i}$, this action will lead to a {\it negative} utility for the principal since $p_{\bar\omega} > 0$.
    In contrast, under the $(\pi, \hat{p})$, the principal can continue to impose $a_{i,5}$ on the agent, i.e.,  $\hat{a}(s) = a_{i,5} \in A^i$, which satisfies the constraints in Problem (\ref{relexted_signle_contract}).
    This would give zero utility to the principal.
    \item If $a(s) \in A^i\setminus \{a_{i,5}\}$,  by the construction, we know that a non-dummy outcome $\hat{\omega}$ is induced with probability $y^s_i$ while $\bar{\omega}$ is induced with probability $1-y^s_i$.
    By the first constraint in Problem (\ref{relexted_signle_contract}), we have that by comparing with action $a_{i,5} \in A^i$,
    \[
    y^s_i p_{\hat{\omega}} + (1-y^s_i) p_{\bar{\omega}} -\frac{1}{2} \ge p_{\bar{\omega}},
    \]
    implying that $y^s_i p_{\hat{\omega}}  -\frac{1}{2} \ge y^s_i p_{\bar{\omega}} > 0$. Furthermore, we know that for any other action $a' \in A^i\setminus \{a_{i,5}\}$, it leads to another non-dummy outcome  $\omega'$ with the same probability $y^s_i$ and the dummy outcome $\bar{\omega}$ with  probability $1-y^s_i$. By the first constraint in Problem (\ref{relexted_signle_contract}), we have 
    \[
    y^s_i p_{\hat{\omega}} + (1-y^s_i) p_{\bar{\omega}} -\frac{1}{2} \ge y^s_i p_{\omega'} + (1-y^s_i) p_{\bar{\omega}} -\frac{1}{2},
    \]
    which implies that $y^s_i p_{\hat{\omega}}  -\frac{1}{2} \ge y^s_i p_{\omega'}  -\frac{1}{2}$. Hence, under contract $\hat{p}$, the action $\hat{a}(s) = a(s)$ is  feasible but gives more utility to the principal since less payment is needed, i.e., $\hat{p}_{\bar\omega} = 0$.
\end{itemize}
Similarly, we consider the case of $a(s) \in A^i$ with $i \in \{N+1, N+2, \dots, 2N\}$.
\begin{itemize}
\item If $a(s) = a_{i,5} \in A^i$, 
    the second constraint in Problem (\ref{relexted_signle_contract}) only requires nonnegative utility for the agent. Under $(\pi, \hat{p})$,  setting $\hat{a}(s) = a(s)$ obviously satisfies the second constraint in Problem (\ref{relexted_signle_contract}) due to $p\ge 0$, which will benefit more the  principal since the expected payment is less. 
    \item If $a(s) \in A^i \setminus \{a_{i,5}\}$, the principal gains negative utility since $p_{\bar{\omega}}>0$. By comparing with action $a_{i,5}$, the first constraint in Problem (\ref{relexted_signle_contract}) implies  
\begin{equation}\label{a5asinplus1to2n}
    p_{\bar{\omega}} -\frac{1}{2} \ge y^s_i p_{\hat{\omega}} + (1-y^s_i) p_{\bar{\omega}},
    \end{equation}
    where the right-hand side is due to  that action $a_{i,5}$ leads to one non-dummy outcome $\hat{\omega}$ with probability $y^s_i$ and to dummy outcome $\bar{\omega}$ with probability $1-y^s_i$. (\ref{a5asinplus1to2n}) further implies $y^s_i p_{\bar{\omega}} -\frac{1}{2} \ge y^s_i p_{\hat{\omega}} \ge 0$, which means  $p_{\bar{\omega}} > p_{\hat{\omega}}$. Therefore,  under $(\pi, \hat{p})$, the principal can impose action $\hat{a}(s) = a_{i,5} \in A^i$.  The second constraint is obviously satisfied.  Under $(\hat{p}, \hat{a}(s))$, the principal gains more utility since the expected payment is less, but the expected reward increases. 
\end{itemize}
Finally, we have that $(\hat{p}, \hat{a})$ gives more utility to the principal, contradicting the optimality of $({p}, {a})$. This conclude that $p_{\bar\omega} = 0$.
\end{proof}

Given a scheme $\pi$ and $(p, a)$ satisfying Lemma \ref{optimalpbraomegaeq0},  our next Lemma shows that full information revelation $I$ (with contract and action profile built upon $(p, a)$)  benefits the principal the most. Denote the set of signals in $I$ as $I_{\Sigma} = \{s^{\theta_1}, s^{\theta_2}, \dots, s^{\theta_{2N}}\}$ such that $I(s^{\theta_i}|\theta_i) =1$.

\begin{lemma}\label{fullisbetterlmma}
    Full information revelation $I$ is an optimal scheme for Problem (\ref{relexted_signle_contract}).
\end{lemma}

\begin{proof}
Assume that the signaling scheme $\pi$ is optimal, and conditional on $\pi$, $(p, a)$ maximizes Problem (\ref{relexted_signle_contract}) thus satisfies Lemma \ref{optimalpbraomegaeq0}. We show that with full information revelation, $(I, p, \hat{a})$ gives more utility to the principal where the action profile $\hat{a}$ is feasible for Problem (\ref{relexted_signle_contract}).

By Lemma \ref{lemamdirectfor_relsaxed_proble}, we can consider a direct mechanism $(\pi, p, a)$ for Problem (\ref{relexted_signle_contract}). For each action set $A^i$, by construction, there will be at most two signals respectively inducing two actions in $A^i$: 1) For $i \in \{1, 2, \dots, N\}$, action $a_{i,5} \in A^i$ and action $a_{i,k} \in A^{i}\setminus \{a_5\}$ where with $ \bar{k} =\arg\max_j p_{\delta_j(i)}$,  $k = \bar{k}$ if $p_{\bar{k}} \ge p_i$ otherwise $k=4$, i.e., the action induces the largest payments; and 2) For $i \in \{N+1, N+2, \dots, 2N\}$, action $a_{i,5} \in A^i$ and action $a_{i,1} \in A^i$. 

For each signal $s$, it is important to  note that to maximize the principal utility, the imposed action $a(s)$ will give nonnegative utility to the principal, since otherwise, the principal can impose action $a_{i,5} \in A^i$ with $i \in \{1, 2, \dots, N\}$ that gives $0$ utility.

First, assume that the principal will impose actions from set  $A^i$ with $i \in \{1, 2, \dots, N\}$ under the mechanism $(\pi, p)$. Let  $s$ and $\bar{s}$ be signals corresponding to action $a(s) \in A^i \setminus\{a_{i,5}\}$ and $a(\bar{s})  =a_{i,5} \in A^i$. Note that in this case, we only need to consider state $\theta_i$ since only  $\theta_i$ can induce positive utility to the principal.
\begin{itemize}
    \item If the signaling scheme $\pi$ only send signal $\bar{s}$,  the principal only get $0$ utility from state $\theta_i$.  Under full information revelation,  the principal can continue to set $\hat{a}(s^{\theta_i}) = a_{i,5} \in A^i$, which still gives $0$ utility.
    \item If $\pi$ only send signal $s$ and $a(s) \in A^i \setminus\{a_{i,5}\}$, by the first constraint in Problem (\ref{relexted_signle_contract}) and $p_{\bar{\omega}}=0$ thanks to Lemma \ref{optimalpbraomegaeq0}, we have 
    \begin{align*}
        y^s_i p_{\omega(a(s))} -\frac{1}{2} &\ge y^s_i p_{\omega(a')} - \frac{1}{2}, \quad \forall a' \in A^{i} \setminus \{a_{i,5}\} \\
        y^s_i p_{\omega(a(s))} -\frac{1}{2} &\ge 0
    \end{align*}
    where $\omega(a)$ is the non-dummy outcome induces by action $a$ and we recall that $y^s$ is the posterior under signal $s$. Under scheme $I$, by setting $\hat{a}(s^{{\theta}_i}) = a(s)$, we will continue to have the first constraint in Problem (\ref{relexted_signle_contract}) hold. The epxected utility that the principal gets from signal $s$ under $(\pi, p)$ is 
\begin{align*}
     \prob{s} \sum_{\theta} \prob{\theta|s} F^{\theta}_{a(s)} \cdot (r - p)
    =\quad & \prob{s} F^s_{a(s)} (r-p) \\
    =\quad & \bar{x} \cdot y_i (1-p_{\omega(a(s))}) \\ 
    =\quad & \bar{x} \cdot \frac{x^s_i}{\bar{x}} (1-p_{\omega(a(s))}) \\
    = \quad & x^s_i (1-p_{\omega(a(s))}) \\
    \le \quad & \frac{1}{2N} (1-p_{\omega(a(s))}) 
\end{align*}
where the inequality follows $x^s_i \le \mu(\theta_i) =\frac{1}{2N}$ and $p_{\omega(a(s))} \le 1$ due to that the principal will not impose an action leading to negative utility.  The right-hand side of the inequality is the utility the principal gains from state $\theta_i$ under scheme $I$. Therefore,  full information revelation benefits the principal more.

\item If $\pi$ sends both of two signals $s$ and $\bar{s}$, then by the above arguments, we can set $\hat{a}(s^{{\theta}_i}) = a(s)$ under scheme $I$, which is still feasible. This will give more utility to the principal.
\end{itemize}

Similarly, we consider the action set $A^{i}$ with $i \in \{N+1, N+2, \dots,  2N \}$. Actually, we know that under any signal, the action $a(s) =a_{i,1}$ will not be imposed by the principal due to $p_{\bar{\omega}} = 0$, which will lead to negative utility to the agent. Hence, only $a(s) = a_{i,5} \in A^{i}$ will be imposed. Similarly, the expected utility that the principal gets from $s$ is
\begin{align*}
     \prob{s} \sum_{\theta} \prob{\theta|s} F^{\theta}_{a_{i,5}} \cdot (r - p)
    =\quad & \prob{s} F^s_{a_{i,5}} (r-p) \\
    =\quad & \bar{x} \cdot y^s_i (1-p_{\omega(a_{i,5})}) \\ 
    =\quad & \bar{x} \cdot \frac{x^s_i}{\bar{x}} (1-p_{\omega(a_{i,5})}) \\
    = \quad & x^s_i (1-p_{\omega(a_{i,5})})   \le \frac{1}{2N} (1-p_{\omega(a_{i,5})}) 
\end{align*}
where the right-hand side of the inequality is the utility obtained from $\theta_i$ under full revelation. Hence, we set $\hat{a}(s^{\theta_i}) = a_{i,5}$.

Finally, we need to consider the cases where there exists some action set $A^i$, whose actions are never imposed by the principal under $\pi$.  This implies that the principal gets zero utility from state $\theta_i$.
Under $(I, p)$, for each of these states $\theta_i$, the principal can impose the opt-out action $\hat{a}(s^{{\theta}_i}) = a_{k,5} \in A^k$ for any $k \in \{1, 2, \dots, N\}$. This will give  zero utility to the principal.

Hence, by the above argument, we can conclude that the principal's  utility from $(I, p, \hat{a})$ is larger than that in $(\pi, p, a)$, which is a contradiction. This concludes that $I$ is an optimal strategy. 
\end{proof}

Lemma \ref{fullisbetterlmma} implies that given any scheme $\pi$ together with contract-action profile $(p, a)$ for Problem (\ref{relexted_signle_contract}), one can always construct a new profile under $I$ such that higher principal utility is achieved. Hence, considering full information revelation $I$ is sufficient for Problem (\ref{relexted_signle_contract}). 

Finally, we show that under $I$, the optimal contract-action profile for Problem (\ref{relexted_signle_contract}) coincides with that of the original Problem \ref{original_problem_contract_single_exp}. By \cite{guruganesh2021contracts}, 
we denote the optimal contract and action profile for Problem \ref{original_problem_contract_single_exp} under $I$ as $(p^*, a^*)$.

\begin{lemma}\label{underfullp8anadori}
Under $I$, Problems (\ref{relexted_signle_contract}) and \ref{original_problem_contract_single_exp} have the same optimal contract-action profile.
\end{lemma}
\begin{proof}
Under $I$, $(p^*, a^*)$ obviously is feasible for Problem (\ref{relexted_signle_contract}). Next, we show that $(p^*, a^*)$ is optimal. 

Conditional on $I$, we maximize Problem (\ref{relexted_signle_contract}) and obtain a solution $(p, a)$. Note that $p_{\bar{\omega}} = 0$ by Lemma \ref{optimalpbraomegaeq0}. We first show that $p\le 1$ holds. Assume by contradiction that there exists some outcome $\omega$ such that $p_{\omega}>1$. We construct a new contract $\bar{p}$ such that if $\bar{p}_{\omega} = p_{\omega}$ for $p_{\omega} \le 1$ and $\bar{p}_{\omega} = 1$ for $p_{\omega} > 1$.
\begin{itemize}
    \item For $i\in \{N+1, N+2, \dots, 2N\}$, under signal $s^{{\theta}_i}$,  we know that either $a(s^{{\theta}_i}) =a_{i,5} \in A^i$ or $a(s^{{\theta}_i}) =a_{i,5} \in A^k$ for some $k\neq i$ holds. If the latter case holds, it will give $0$ utility to the principal. $(\bar{p}, a)$ is feasible and gives the same utility to  the principal. If $a(s^{{\theta}_i}) =a_{i,5} \in A^i$, $(\bar{p}, a)$ is obviously also feasible in Problem (\ref{relexted_signle_contract}) and the principal's utility weakly increases.

\item Consider $i\in \{1, 2,\dots, N\}$. If $a(s^{{\theta}_i}) =a_{i,5} \in A^i$, $(\bar{p}, a)$ is feasible to Problem (\ref{relexted_signle_contract}) and gives the principal same utility. If $a(s^{{\theta}_i}) \in A^i \setminus \{a_{i,5}\}$, it is easy to see that $(\bar{p}, a)$ is feasible in Problem (\ref{relexted_signle_contract}) since the cost $c_{a(s^{{\theta}_i})} = \frac{1}{2}$, and the principal's utility weakly increases.
\end{itemize}
Therefore, $p \le 1$ must hold. This implies that the principal can always get nonnegative utility from each state under $I$. This implies that under the scheme $I$, it is not profitable for the principal under any signal $s^{{\theta}_i}$ to impose an action from $A^k$, $k\neq i$, which always induces the dummy outcome. Hence, for each state $\theta_i$, we can constrain the action set for the agent to be $A^i$. Moreover, to maximize the utility, the principal will impose the action that gives the maximum principal utility when satisfying the constraints in Problem (\ref{relexted_signle_contract}). Consider  $i\in \{1, 2, \dots, N\}$. If the imposed action is $a(s^{{\theta}_i}) \in A^i \setminus \{a_{i,5}\}$, this must be the agent's best-response action by the constraints in (\ref{relexted_signle_contract}).  If the imposed action is $a(s^{\theta}_i) = a_{i,5} \in A^i$, which gives zero utility to the principal and agent, this must imply that other actions in $A^i$ give a negative utility to the agent, and thus $a_{i,5}$ is the best-response action. Similarly for $i\in \{N+1, N+2, \dots, 2N\}$, the principal will only impose action $a(s^{\theta}_i) = a_{i,5} \in A^i$, which is also the agent's best-response action. These arguments imply that under $(I, p)$, the imposed action by the principal is indeed the best-response action of the agent. In other words, under scheme $I$, Problem (\ref{relexted_signle_contract}) is equivalent to Problem \ref{original_problem_contract_single_exp}. Therefore, $(p^*, a^*)$ is the optimal solution to Problem (\ref{relexted_signle_contract}).   
\end{proof}


Lemma \ref{fullisbetterlmma} and \ref{underfullp8anadori} imply that $(I, p^*, a^*)$ is the optimal solution to Problem (\ref{relexted_signle_contract}), and since the utility of (\ref{relexted_signle_contract}) upper bounds that of \ref{original_problem_contract_single_exp}, it immediately implies that it is also the optimal solution to Problem \ref{original_problem_contract_single_exp}.

By \cite{guruganesh2021contracts} that it is NP-hard to distinguish if the optimal solution extracts expected profit at least $0.67820248214$ or at most $0.678018005$ in a Bayesian contract design problem, the same result also holds for joint design problem with single explicit contract as we have shown that the full information revelation $I$ is optimal. This concludes the proof.
\end{proof}

\section{Explicit Linear Contracts Are Tractable}

In the previous sections, we showed that it is \NP-Hard to compute explicit (single or menus of) contracts in general. This further motivates our study of simpler classes of contracts.  \citet{bajari2001incentives} argue that ``the vast majority of contracts are variants of simple fixed-price (FP) and cost-plus (C+) contracts". Hence, looking for simple contracts in the joint design problem is significant. 
Moreover, this is a necessary requirement to allow the efficient computation of explicit contracts. 
Here, we focus on the most common class of simple contracts, i.e., linear contract.
This simple class of contracts enjoy two important properties: i) it is usually easier to optimize, and ii) it is ‘‘powerful'' and simple enough to be employed in real-world applications.
Indeed, this simple type of contract is widely seen in practice. For example, a lost wallet notice may simply say the reward as {\it half of the cash in the wallet will be shared with the finder}, which is a linear contract of parameter $\frac{1}{2}$. 

Formally, (menus of) linear contracts are defined as follows:
\begin{definition}
    (Linear Contract) A single contract $p\in \mathbb{R}_+^m$ is linear if it is characterized by only one parameter $\alpha \in [0, 1]$, such that $p=\alpha r$. Similarly, a menu of  explicit liner contracts is characterized by a linear contract $p^s = \alpha(s) r$ for all  $ s\in \Sigma$, where $\alpha(s) \in [0, 1]$ depends on signal $s$.
\end{definition}
In the next sections, we show that the joint design with explicit (menus of) linear contracts is indeed tractable. In particular, it admits an FPTAS. 

\subsection{Single Linear Contract}

If the principal restricts to commit to a single linear contract, Problem \ref{original_problem_contract_single_exp} reads as follows
\begin{equation}\label{original_problem_contract_signal_linear}
\begin{aligned}
    \max_{\pi, \alpha \in [0,1]} \quad & \sum_{\theta \in \Theta} \mu(\theta) \sum_{s \in \Sigma} \pi(s|\theta) (1-\alpha ) \langle F^{\theta}_{a(s)},   r \rangle \\
    \textnormal{s.t.} \quad & \sum_{\theta \in \Theta} \mu(\theta)\pi(s|\theta) \Big[ \alpha \langle F^{\theta}_{a(s)},  r \rangle -c_{a(s)} \Big] \ge \sum_{\theta \in \Theta} \mu(\theta)\pi(s|\theta) \Big[  \alpha\langle F^{\theta}_{i},  r \rangle -c_{i}\Big], \quad  \forall s \in \Sigma, i \in \langle n \rangle \\
& \sum_{s \in \Sigma} \pi(s|\theta) = 1, \pi(s|\theta) \ge 0, \quad \forall \theta \in \Theta, s \in \Sigma
\end{aligned}
\end{equation}
By arguments similar to Proposition \ref{singlcontractdirectmech}, we can show that the problem admits a maximum and that we can restrict to consider direct mechanisms. The algorithm is depicted in Algorithm \ref{alg_single_linear}. 
Intuitively, it discretizes the set of possible linear contracts with a grid of size $\epsilon$. Then, once the contract is fixed, it is possible to compute the optimal signaling scheme associated with the contract in polynomial time. The main challenge is to show that we do not lose too much utility by considering only a grid of possible contracts.
We can prove the following guarantees for our algorithm:

\begin{algorithm}[t]
\caption{FPTAS for joint design problem with single explicit linear contract}\label{alg_single_linear}
\textbf{input:}{ parameter $\epsilon>0$}\;
$\alpha^* \gets 0$\;
$U^* \gets 0$ \;
\For{$\alpha = 0, \eps, 2\eps, \dots, 1$}{
  $U \gets$ {utility from solving Problem (\ref{original_problem_contract_signal_linear}) with linear contract} $\alpha$\;
  \If{$U > U^*$}{
      $\alpha^* \gets \alpha$\;
      $U^*\gets U$
    }
}
$\pi^* \gets$ {signaling scheme from solving Problem (\ref{original_problem_contract_signal_linear}) with linear contract} $\alpha^*$\;
\Return $(\pi^*, \alpha^*)$\;
\end{algorithm}

\begin{proposition}\label{proposition61linearsignal}
    For any $\eps>0$, there exists an algorithm running in time polynomial in the instance size and $\frac{1}{\eps}$ that returns a solution to Problem~\eqref{original_problem_contract_signal_linear} with principal's utility  at least $\opt -\eps$, where $\opt$ is the optimal principal's utility.
\end{proposition}
\begin{proof}
    Suppose that the optimal signaling scheme is $\bar \pi$, and $\bar \alpha$ be the optimal linear contract.  Given signal $s$, let 
    \[f_s(\alpha) = \frac{\sum_{\theta \in \Theta} \mu(\theta)  \bar \pi(s|\theta) (1-\alpha ) \langle F^{\theta}_{a(\alpha)},   r \rangle}{\sum_{\theta \in \Theta} \mu(\theta)  \bar \pi(s|\theta)}\] be the principal's utility, where $a(\alpha)$ is the agent's best-response action under contract $\alpha$. Notice that $\sum_{s}\sum_{\theta \in \Theta} \mu(\theta)  \bar \pi(s|\theta) f_{s}(\bar \alpha)= \opt$.
    It is easy to see that $f_s(\alpha)$ is equivalent to a the utility in principal-agent problem in which the action-outcome distribution induced by an action $a$ is given by $\frac{\sum_{\theta \in \Theta} \mu(\theta)  \pi(s|\theta) F^{\theta}_{a}}{\sum_{\theta \in \Theta} \mu(\theta)  \bar \pi(s|\theta)}$.
    It is well-known that $f_s(\alpha)$ is a piece-wise linear upper-semicontinuous function with at most $n+1$ discontinuous points~\cite{dutting2019simple,dutting2023optimal}. Moreover, the slope of all linear functions is at least $-1$ and non-positive.
    Then, consider smallest contract $\alpha^*$ in the grid such that $\alpha^*>\bar \alpha$. Clearly, $\alpha^*-\bar \alpha\le \epsilon$.
    Hence, we have that for each $s$ it holds
    \[f_s(\alpha^*)\ge f_s(\bar \alpha)-\epsilon.\]
    Hence, the utility of mechanism $(\pi^*,\alpha^*)$ is:
    \[ \sum_{s}\sum_{\theta \in \Theta} \mu(\theta)  \bar \pi(s|\theta) f_{s}( \alpha^*) \ge\sum_{s}\sum_{\theta \in \Theta} \mu(\theta)  \bar \pi(s|\theta) f_{s}( \bar \alpha)-\epsilon= \opt-\epsilon.  \]

    This implies that when Algorithm~\ref{alg_single_linear} solves Problem~\eqref{original_problem_contract_signal_linear} with contract $\alpha^*$ (which is an LP), it returns a mechanism with utility at least $\opt-\epsilon$.
    Finally, it is easy to see that Algorithm~\ref{alg_single_linear} runs in time polynomial in the instance size and $1/\epsilon$.
\end{proof}

\subsection{Menus of linear contracts}

In this section, we turn our attention to the menus of linear contracts. The problem of finding the optimal menus of linear contracts can be formulated as Problem~(\ref{mneulinearcontractformulatedion}). 
Notice that the payment scheme $\mathcal{P}$ here is represented by a set of linear contracts $\alpha(s)$, one for each signal.
\begin{equation}\label{mneulinearcontractformulatedion}
\begin{aligned}
    \sup_{\pi, \mathcal{P}} \quad & \sum_{\theta \in \Theta} \mu(\theta) \sum_{s \in \Sigma} \pi(s|\theta) (1-\alpha(s)) \langle F^{\theta}_{a(s)},   r \rangle \\
    \textnormal{s.t.} \quad & \sum_{\theta \in \Theta} \mu(\theta)\pi(s|\theta) \Big[ \alpha(s) \langle F^{\theta}_{a(s)},   r \rangle -c_{a(s)} \Big] \hspace{-0.07cm} \ge \hspace{-0.07cm}  \sum_{\theta \in \Theta} \mu(\theta)\pi(s|\theta) \Big[ \alpha(s) \langle F^{\theta}_{i},   r \rangle -c_{i}\Big] \hspace{0.18cm}  \forall s \in \Sigma, i \in \langle n \rangle \\
& \sum_{s \in \Sigma} \pi(s|\theta) = 1, \pi(s|\theta) \ge 0 \quad \forall \theta \in \Theta, s \in \Sigma
\end{aligned}
\end{equation}

Designing a menu of linear contracts is much more challenging than designing a linear contract. Indeed, as we showed in Proposition~\ref{sumoptimaiontamenuscontr} for arbitrary contracts, direct signaling schemes are not optimal and it is unclear whether an optimal menu exists. Hence, we cannot simply apply an enumeration as in Proposition  \ref{proposition61linearsignal}.
Nonetheless, we show that a polynomial number of signals and contracts is sufficient to approximate the optimal mechanism.
Then, we use this result to design an FPTAS for the problem.

Our algorithm is depicted in Algorithm \ref{alg_menu_linear}. It considers only contracts from a grid of size $1/\epsilon$. To each of this contracts $\alpha$, it assigns a direct ‘‘marginal'' signaling scheme that employs (at most) one signal $s_{\alpha,i}$ for each action $a_i$. 
Hence, there are at most $\frac{n}{\eps}$ signals.
Then, it solves Problem (\ref{mneulinearcontractformulatedion}), where the payment function is fixed. Hence, the problem is linear.

Our proof works in two main steps. First, we show that we are not loosing too much utility considering  only a grid of possible linear contracts. Then, we observe that signals that induce the same payment and best-response can be joined since, once the payment is fixed, the best response region is convex.
 Below, we formally introduce our result.

\begin{algorithm}[t]
\caption{FPTAS for joint design problem with a menu of explicit linear contracts}\label{alg_menu_linear}
\textbf{input:}{parameter $\epsilon>0$}\;
initialize set of signals $\Sigma\gets \emptyset$\;
\For{$\alpha= 0,\eps,2\epsilon, \ldots,1$}{
\For{$i \in \langle n \rangle $}{
    $\Sigma\gets \Sigma \cap \{s_{\alpha,i}\}$\;
    $\alpha(s_{\alpha,i})\gets \alpha$\Comment*[r]{assign linear contract $\alpha$ to signal $s_{\alpha,i,}$} 
}}
$\pi^* \gets$  signaling scheme from solving Problem (\ref{mneulinearcontractformulatedion}) with fixed $\alpha(\cdot)$ and $a(s_{\alpha, i})=a_i$ \;
\Return $(\pi^*, \alpha(\cdot))$
\end{algorithm}

\begin{proposition}
    For any $\eps>0$, there exists an algorithm running in time polynomial in the instance size and $\frac{1}{\eps}$ that returns a solution to Problem~\eqref{mneulinearcontractformulatedion} with principal's utility at least $\opt-\eps$, where $\opt$ is the optimal principal's utility.
\end{proposition}
\begin{proof}
    Consider a signaling scheme $\bar \pi$ with set of signals $\bar \Sigma$ and a menu of linear contracts $\bar \alpha(s)$.
    For each $s \in \bar \Sigma$ and $\alpha \in [0,1]$ define 
    \[f_s(\alpha) = \frac{\sum_{\theta \in \Theta} \mu(\theta) \bar \pi(s|\theta) (1-\alpha) \langle F^{\theta}_{a(\alpha)},   r \rangle}{\sum_{\theta \in \Theta} \mu(\theta) \bar \pi(s|\theta)}\] as the optimal utility for contract $\alpha$, where $a(\alpha)$ is the agent's best-response action under contract $\alpha$ and signal $s$. By a similar argument as the one in Proposition \ref{proposition61linearsignal}, we know that $f_s(\alpha)$ is a piece-wise linear upper-semicontinuous function with at most $n+1$ discontinuous points. Moreover, the slope of all linear functions is at least $-1$.

    For any signal $s$, let $\alpha^*(s)$ be the smallest contract in the grid such that $\alpha^*(s)>\bar \alpha(s)$. Clearly, $\alpha^*-\bar \alpha\le \epsilon$.
    Hence, we have that for each $s$ it holds
    \[f_s(\alpha^*(s))\ge f_s(\bar \alpha(s))-\epsilon.\]

    Hence, the utility of the mechanism $\bar\pi, \alpha^*(\cdot)$ is:
    \[ \sum_{s}\sum_{\theta \in \Theta} \mu(\theta)  \bar \pi(s|\theta) f_{s}( \alpha^*(s)) \ge\sum_{s}\sum_{\theta \in \Theta} \mu(\theta)  \bar \pi(s|\theta) f_{s}( \bar\alpha(s))-\epsilon= \opt-\epsilon.  \]

    Now, we want to show that the principal's utility under mechanism $\bar\pi, \alpha^*(\cdot)$ is upperbounded by a mechanism in which the set of signals is $\Sigma=\{ s_{\alpha,i}:\alpha \in \{0,\epsilon,\ldots,1\}, i \in \langle n \rangle\}$, and such that $\alpha^*(s_{\alpha,i})=\alpha$.
    Indeed, consider two signals $s_1,s_2\in \bar \Sigma$ such that $ \alpha^*(s_1)=\alpha^*(s_2)$, and that induce the same best response $a$ under $(\bar \pi,\alpha^*(\cdot))$. Consider a new signaling scheme $\pi^*$ in which we replace the two signals  $s_1,s_2$ with one new signal $\bar{s}$ by setting $\pi^*(\bar{s}|\theta) = \bar{\pi}(s_1|\theta) + \bar{\pi}(s_2|\theta)$ for all $\theta \in \Theta$, and let $\alpha^*(\bar s)=\alpha^*(s_1)$.
    By linearity, we can easily verify that $a$ is a best response to $\bar{s}$ and and the principal's utility does not change, i.e.,
    \begin{small}
\begin{equation*}
\sum_{\theta \in \Theta} \mu(\theta)  \pi^*(\bar{s}|\theta) (1-\alpha^*(\bar{s})) \langle F^{\theta}_{a},   r \rangle =  \sum_{\theta \in \Theta} \mu(\theta)  \bar \pi(s_1|\theta) (1-{\alpha^*}(s_1)) \langle F^{\theta}_{a},   r \rangle  + \sum_{\theta \in \Theta} \mu(\theta)  \bar \pi(s_2|\theta) (1-{\alpha^*}(s_2)) \langle F^{\theta}_{a}, r \rangle 
\end{equation*}
\end{small}
    
    Applying this procedure to each couple of signals, we obtain as set of signals $\Sigma$.
    Hence, we proved that there exists a mechanisms $\pi^*, \alpha^*(\cdot)$ that
    \begin{itemize}
        \item employs as set of signals $\Sigma$,
        \item has $\alpha^*(s_{\alpha,i})=\alpha$ {and the best response is $a(s_{\alpha,i}) = i$} for each $s_{\alpha,i}\in \Sigma$,
        \item achieves principal's utility at least $\opt-\epsilon$.
    \end{itemize}

Finally, we can formulate the problem of obtaining the optimal mechanism in this set as:
\begin{equation}\label{menulinearapproximate}
\begin{aligned}
    \max_{\pi} \quad & \sum_{\alpha \in \mathcal{A}} \sum_{i \in \langle n\rangle} \sum_{\theta \in \Theta} \mu(\theta) \pi(s_{\alpha, i}|\theta) (1-\alpha) \langle F^{\theta}_{i},   r \rangle \\
    \textnormal{s.t} \quad & \sum_{\theta \in \Theta} \mu(\theta)\pi(s_{\alpha, i}|\theta) \Big[ \alpha \langle F^{\theta}_{i},   r \rangle -c_{i} \Big] \ge \sum_{\theta \in \Theta} \mu(\theta)\pi(s_{\alpha, i}|\theta) \Big[ \alpha \langle F^{\theta}_{j},   r \rangle -c_{j}\Big] \  \forall s_{\alpha, i} \in \Sigma, j \in \langle n \rangle \\
& \sum_{s \in \Sigma} \pi(s|\theta) = 1, \pi(s|\theta) \ge 0 \quad \forall \theta \in \Theta, s \in \Sigma
\end{aligned}
\end{equation}
Hence, we can conclude that Program~\eqref{menulinearapproximate} returns a mechanism with utility $\opt-\epsilon$.
Moreover, it is an LP and hence can be solved in time polynomial time in the instance size and $1/\epsilon$.
This concludes the proof.
\end{proof}

\section{Conclusions}


In this paper, we initiate the problem of jointly designing information and contracts, which captures many applications in practical life. By differentiating the principal's flexibility in designing contracts, we consider ambiguous contracts, menus of explicit contracts and single explicit contract. Our main technical contributions are devoted to understanding the computational aspects of the  problem. 

We believe this newly proposed problem can enrich the literature by simultaneously considering information asymmetry and moral hazard. There are many research questions in this direction. Below, we discuss two of them that could potentially interest the community.
\begin{itemize}
    \item One interesting question left in this paper is about {\it how many signals are enough for an optimal design in \menu}. This would help us better understand the structure of optimal signaling schemes. We suspect that an exponential number of signals is needed for the optimal design. If this is the case, a follow-up question may be: {\it is it possible to derive good approximations  for some simple class of signaling scheme such as the ones with a limited number of signals?}. For example, as in \cite{bergemann2022selling}, it is interesting to understand the approximation of the full-information scheme. Note that the result in \cite{dutting2019simple} implies that an $O(n)$ approximation can be derived for full-information revelation. However, it is unclear whether it is possible to derive a constant approximation.
    \item Another interesting direction is the extension of our model to the setting of multiple agents. For example, in a lost item notice, the owner may want to incentivize more agents to look for the item so that the probability of getting back the item is higher. Although in general this problem could potentially be \NP-hard given our current results, it would be interesting to consider classes of simple agents, e.g.,  with binary actions. Moreover, as in Bayesian persuasion, one can further consider the settings where the principal could communicate with the agent privately (i.e., private persuasion~\cite{arieli2019private,dughmi2017algorithmicnoex}) and publicly (i.e., public persuasion~\cite{castiglioni2023public,xu2020tractability}).
\end{itemize}

\newpage
\bibliographystyle{ACM-Reference-Format}
\bibliography{refer}
\newpage 
\appendix


\section{Proofs Omitted from Section~\ref{sec:ambigous}}

\subsection{Proof of Lemma~\ref{lemma_recommenddirect}}
 Given a feasible solution $(\tilde{\pi}, \tilde{\mathcal{P}})$ to Problem \ref{original_problem_contract_signal}, we will construct a direct mechanism $(\pi, \mathcal{P})$ that achieves at least the same principal utility. 
Suppose that in the feasible  solution $(\tilde{\pi}, \tilde{\mathcal{P}})$, under two signals $s_1$ and $s_2$, the agent optimally chooses the same action $a \in \langle n\rangle$, i.e., $a(s_1)=a(s_2)=a$. Hence, we have the following two IC constraints for two signals, 
    \begin{gather*}
        \sum_{\theta} \mu(\theta) \tilde{\pi}(s_1|\theta) [\langle F^{\theta}_{a}, \tilde{p}^{s_1, \theta} \rangle -c_{a}] \ge \sum_{\theta} \mu(\theta) \tilde{\pi}(s_1|\theta) [\langle F^{\theta}_{i},  \tilde{p}^{s_1, \theta} \rangle -c_{i}], \quad \forall i \in \langle n \rangle \\
        \sum_{\theta} \mu(\theta) \tilde{\pi}(s_2|\theta) [\langle F^{\theta}_{a},  \tilde{p}^{s_2, \theta} \rangle-c_{a}] \ge \sum_{\theta} \mu(\theta) \tilde{\pi}(s_2|\theta) [\langle F^{\theta}_{i}, \tilde{p}^{s_2, \theta} \rangle -c_{i}], \quad \forall i \in \langle n \rangle
    \end{gather*}
and the principal's utility from these two signals as 
\[
\sum_{\theta \in \Theta} \mu(\theta) \tilde{\pi}(s_1|\theta) \langle F^{\theta}_{a}, r - \tilde{p}^{s_1,\theta} \rangle + \sum_{\theta \in \Theta} \mu(\theta) \tilde{\pi}(s_2|\theta) \langle F^{\theta}_{a},   r - \tilde{p}^{s_2,\theta}\rangle
\]
We construct a direct mechanism $(\pi, \mathcal{P})$, where we replace $s_1$ and $s_2$ with a new signal $s$ such that $\pi(s|\theta) = \tilde{\pi}(s_1|\theta)+\tilde{\pi}(s_2|\theta)$ and $p^{s, \theta} = \frac{\tilde{\pi}(s_1|\theta)}{\pi(s|\theta)}\tilde{p}^{s_1, \theta} + \frac{\tilde{\pi}(s_2|\theta)}{\pi(s|\theta)} \tilde{p}^{s_2, \theta}$. We verify that under the new mechanism, the optimal action for the agent under signal $s$ and contracts $\{p^{s, \theta}\}_{\theta\in \Theta}$ is still $a$, i.e., for any $i \in \langle n \rangle$, we have
\begin{align*}
 \sum_{\theta} \mu(\theta) {\pi}(s|\theta) [\langle F^{\theta}_{a},  {p}^{s_1, \theta} \rangle -c_{a}] &= \sum_{\theta} \mu(\theta) {\pi}(s|\theta) [\langle F^{\theta}_{a},   \frac{\tilde{\pi}(s_1|\theta)}{\pi(s|\theta)}\tilde{p}^{s_1, \theta} + \frac{\tilde{\pi}(s_2|\theta)}{\pi(s|\theta)} \tilde{p}^{s_2, \theta}  \rangle -c_{a}] \\
 & = \sum_{\theta} \mu(\theta) \tilde{\pi}(s_1|\theta) [\langle F^{\theta}_{a},  \tilde{p}^{s_1, \theta} \rangle -c_{a}] + \sum_{\theta} \mu(\theta) \tilde{\pi}(s_2|\theta) [\langle F^{\theta}_{a},  \tilde{p}^{s_2, \theta} \rangle -c_{a}]\\
 & \ge \sum_{\theta} \mu(\theta) \tilde{\pi}(s_1|\theta) [\langle F^{\theta}_{i},  \tilde{p}^{s_1, \theta} \rangle -c_{i}] + \sum_{\theta} \mu(\theta) \tilde{\pi}(s_2|\theta) [\langle F^{\theta}_{i},  \tilde{p}^{s_2, \theta} \rangle -c_{i}] \\
 & = \sum_{\theta} \mu(\theta) {\pi}(s|\theta) [\langle F^{\theta}_{i},  {p}^{s_1, \theta} \rangle -c_{i}]
\end{align*}
The revenue of the principal remains the same, i.e., 
\[
\sum_{\theta \in \Theta} \mu(\theta) {\pi}(s_1|\theta) \langle F^{\theta}_{a},   r - {p}^{s_1,\theta} \rangle = \sum_{\theta \in \Theta} \mu(\theta) \tilde{\pi}(s_1|\theta) \langle F^{\theta}_{a},   r - \tilde{p}^{s_1,\theta} \rangle + \sum_{\theta \in \Theta} \mu(\theta) \tilde{\pi}(s_2|\theta) \langle F^{\theta}_{a},   r - \tilde{p}^{s_2,\theta} \rangle
\]
Hence, for signal $s$, we directly recommend action $a$. This concludes the proof.

\subsection{Proof of Proposition~\ref{theoremsupereconmanntoacdhi} }

The proposition is proved by constructing one instance where the maximum cannot be achieved.

{\bf Construction of the Instance.} There are three states $\theta_1, \theta_2$ and $\theta_3$, and the  probabilities of states are $\mu(\theta_i)=\frac{1}{3}$ for $i=1,2,3$. There are $3$ actions $\mathcal{A} = \{a_1, a_2, a_3\}$ and $4$ outcomes $\Omega = \{ \omega_1, \omega_2, \omega_3, \omega_4 \}$. The rewards for outcomes are $r_1 = 1, r_2 = \frac{1}{2}, r_3=0, r_4=0$, while the costs for actions are $c_1=0, c_2=0, c_3=\frac{1}{8}$. Next, we construct the probability matrices $F^{\theta}$ (also depicted in Figure \ref{proceduremultisingleprocess}) 
\begin{itemize}
    \item For state $\theta_1$, we set $F^{\theta_1}_{a_1, \omega_2} = F^{\theta_1}_{a_1, \omega_3}=\frac{1}{2}$, $F^{\theta_1}_{a_2, \omega_4}=1$ and $F^{\theta_1}_{a_3, \omega_2}=1$. 
    \item For state $\theta_2$, we set $F^{\theta_2}_{a_1, \omega_3} =1$, $F^{\theta_2}_{a_2, \omega_1}=1$ and $F^{\theta_2}_{a_3, \omega_4}=1$. 
    \item For state $\theta_3$, we set $F^{\theta_3}_{a_1, \omega_1} =F^{\theta_3}_{a_1, \omega_3} =\frac{1}{2}$, $F^{\theta_3}_{a_2, \omega_4}=1$ and $F^{\theta_3}_{a_3, \omega_1}=1$. 
\end{itemize}
\begin{figure}[h]
\centering
\includegraphics[scale=0.4]{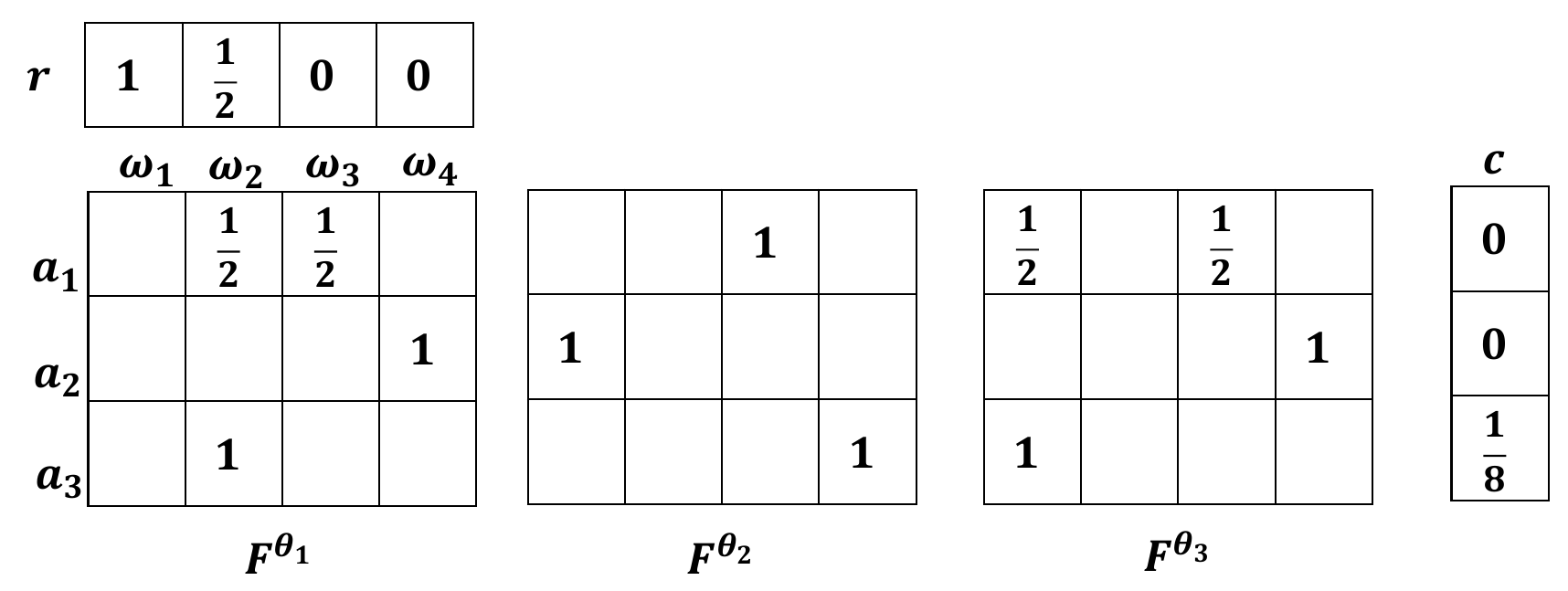}
\caption{The construction of three probability matrices in the constructed instance of Theorem \ref{theoremsupereconmanntoacdhi}, where the blank entries is $0$.}
\label{proceduremultisingleprocess}
\end{figure}

{\bf The supremum \texttt{SUP} is at least $\frac{9}{12}$.} To prove it, we only need to show that 
there exists one solution achieving principal utility as $\frac{9}{12} - \frac{9\eps}{8}$ where $\eps>0$ is an arbitrarily small constant. 

We construct a signaling scheme with $3$ signals $\Sigma = \{s_1, s_2, s_3\}$.  For signal $s_1$, let $\pi(s_1 | \theta)=0$ for all $\theta$. For signal $s_2$, let $\pi(s_2|\theta_1) = \pi(s_2|\theta_3) = 0 $ and $\pi(s_2|\theta_2) = 1-3\eps$. Let contract be $p^{s_2, \theta_2} = [0, 0, 0, 0]$. For signal $s_3$, let $\pi(s_3|\theta_1) = \pi(s_3|\theta_3) = 1 $ and $\pi(s_3|\theta_2) = 3\eps$. 
The contracts are defined as $p^{s_3, \theta_1} = p^{s_3, \theta_3} = [0, 0, 0, 0]$ and $p^{s_3, \theta_2} = [0, 0, 0, \frac{1}{12\eps} + \frac{1}{8}]$.

By construction, the principal will not send signal $s_1$. Under signal $s_2$, the agent can infer that the state is $\theta_2$ and the contract is $p^{s_2, \theta_2}$. Hence, it is optimal for the agent to choose action $a_2$ by tie-breaking. Under signal $s_3$, we show that it is optimal for the agent to take action $a_3$. If the agent chooses action $a_3$, the payoff (by removing the normalization term) is
\begin{align*}
    \sum_{\theta} \mu(\theta) \pi(s|\theta) [\langle F^{\theta}_{a_3}, p^{s, \theta}\rangle -c_{3}] = \frac{1}{3}\times 1\times (0-\frac{1}{8}) + \frac{1}{3}\times 1\times (0-\frac{1}{8}) + \frac{1}{3}\times 3\eps \times (\frac{1}{12\eps} + \frac{1}{8} - \frac{1}{8})= 0
\end{align*}
and the principal's utility is 
\begin{align*}
    \sum_{\theta} \mu(\theta) \pi(s_3|\theta) \langle F^{\theta}_{a_3},  r - {p}^{s_3, \theta} \rangle& =   \frac{1}{3} \times  1  \times \frac{1}{2} + \frac{1}{3} \times  1  \times 1 + \frac{1}{3} \times  3\eps  \times (0 - \frac{1}{12\eps} - \frac{1}{8}) \\
    & =\frac{1}{3} \Big[ 1 + \frac{1}{4}  -\frac{3\eps}{8} \Big]\\
    & = \frac{5}{12} - \frac{\eps}{8}
\end{align*}
If the agent chooses $a_2$, the agent's payoff is $0$, but the principal's utility is $\sum_{\theta} \mu(\theta) \pi(s_3|\theta) \langle F^{\theta}_{a_3}, r - {p}^{s_3, \theta}  \rangle = \eps< \frac{5}{12} - \frac{\eps}{8}$. If the agent chooses $a_1$, the agent's payoff is $0 $, but the principal's utility is \begin{align*}
    \sum_{\theta} \mu(\theta) \pi(s_3|\theta) \langle F^{\theta}_{a_1}, r - {p}^{s_3, \theta} \rangle & =   \frac{1}{3} \times  1  \times \frac{1}{2}\times \frac{1}{2} + \frac{1}{3} \times  1  \times \frac{1}{2} + \frac{1}{3} \times  3\eps  \times 0 \\
    & =\frac{1}{3} \Big[ \frac{1}{4} +\frac{1}{2} \Big]\\
    & = \frac{1}{4} < \frac{5}{12} - \frac{\eps}{8}
\end{align*}
Hence, by tie-breaking rule, the agent will choose action $a_3$. Therefore,
the principal's total utility is 
\begin{align*}
    \sum_{s \in \Sigma} \sum_{\theta} \mu(\theta) \pi(s|\theta) \langle F^{\theta}_{a(s)}, r - {p}^{s, \theta} \rangle& = \Big[ \frac{1}{3} \times (1-3\eps) \times 1 \Big] + \frac{5}{12} - \frac{\eps}{8}  = \frac{9}{12} - \frac{9\eps}{8}
\end{align*}

{\bf Any feasible mechanism has utility strictly less than $\frac{9}{12}$.}
By Lemma \ref{lemma_recommenddirect}, it is without loss of generality to consider a signaling scheme with at most $3$ signals corresponding to actions. If signal $a_2$ is sent with positive probability, i.e.,  $\prob{a_2}>0$, then it should be optimal for the agent to take action $a_2$ by Lemma \ref{lemma_recommenddirect}. For any feasible solution, signal $a_2$   induces the principal utility at most 
\begin{equation}\label{actiona2optimal}
 \sum_{\theta} \mu(\theta) \pi(a_2|\theta) \langle F^{\theta}_{a_2},   r - {p}^{a_2, \theta} \rangle \le \mu(\theta_2) \pi(a_2|\theta_2) 
\end{equation}
which can be achieved by setting $\pi(a_2|\theta_1) = \pi(a_2|\theta_3) = 0$  and contract $p^{a_2, \theta_2} = [0, 0, 0, 0]$. Similarly, for signal $a_1$ and the agent chooses action $a_1$, the principal's utility is at most 
\begin{equation}\label{actiona1signal}
\sum_{\theta} \mu(\theta) \pi(a_1|\theta) \langle F^{\theta}_{a_1},  r - {p}^{a_1, \theta} \rangle \le \frac{1}{4}\mu(\theta_1)\pi(a_1|\theta_1) + \frac{1}{2}\mu(\theta_3)\pi(a_1|\theta_3)
\end{equation}
which can be achieved by setting $\pi(a_1|\theta_2) =  0$ and contract $p^{a_1, \theta_1} = p^{a_1, \theta_3} = [0, 0, 0, 0]$.

By (\ref{actiona2optimal}), to maximize the principal's utility when sending signal $a_2$, it must hold that $\pi(a_2|\theta_1) = \pi(a_2|\theta_3) = 0$. Hence, the following discussion is divided into two cases: 1) $\pi(a_2|\theta_2) = 1$ and 2)  $\pi(a_2|\theta_2) < 1$. 

{\bf Case $\pi(a_2|\theta_2) = 1$.} By previous discussion, we know that in the optimal solution, we must have $\pi(a_2|\theta_1) = \pi(a_2|\theta_3) = 0$ since $a_2$ induces no utility to the principal under states $\theta_1$ and $\theta_3$. Also, we have $\pi(a_3|\theta_2) = 0$. Now, we consider signal $a_3$. 

\begin{itemize}[leftmargin=5.5mm]
    \item If we have $\prob{a_3} = 0$, i.e., action $a_3$ will not be induced, then by the previous discussion, it is optimal to set $\pi(a_1|\theta_1) = \pi(a_1|\theta_3) =1$ with {\it zero} contracts. Hence, the maximum principal utility is 
\[
\frac{1}{4}\mu(\theta_1)\pi(a_1|\theta_1) + \frac{1}{2}\mu(\theta_3)\pi(a_1|\theta_3) + \mu(\theta_2) \pi(a_2|\theta_2) = \frac{1}{3}\times  (1+\frac{3}{4}) < \frac{9}{12}
\]
\item If $\prob{a_3} > 0$, i.e., the agent will choose action $a_3$ if recommended, then by the IC constraint for action $a_3$, the agent's payoff should be greater than that by deviating to $a_1$,
\begin{equation}\label{eq_a3toa1_ic}
\begin{aligned}
& \mu(\theta_1) \pi(a_3|\theta_1) [\langle F^{\theta_1}_{a_3}, p^{a_3, \theta_1} \rangle -c_{3}] +  \mu(\theta_3) \pi(a_3|\theta_3) [\langle F^{\theta_3}_{a_3},  p^{a_3, \theta_3}\rangle  -c_{3}] \\
& \quad \ge \mu(\theta_1) \pi(a_3|\theta_1) [\langle F^{\theta_1}_{a_1},  p^{a_3, \theta_1}\rangle  -c_{1}] +  \mu(\theta_3) \pi(a_3|\theta_3) [\langle F^{\theta_3}_{a_1},  p^{a_3, \theta_3} \rangle  -c_{1}] 
\end{aligned}
\end{equation}
By the construction of matrices $F^{\theta_1}$ and $F^{\theta_3}$, it is optimal to set $p^{a_3, \theta_1} = [0, p^{a_3, \theta_1}_2, 0 , 0]$ and $p^{a_3, \theta_3} = [p^{a_3, \theta_3}_1, 0, 0, 0]$ where $p^{a_3, \theta_1}_2 \ge 0$ and $p^{a_3, \theta_3}_1 \ge 0$. This is because setting positive payments on other outcomes can only decrease the principal's utility and increase the incentives for agent to deviate.

Hence, by (\ref{eq_a3toa1_ic}) we have
\begin{align*}
&\mu(\theta_1) \pi(a_3|\theta_1)  (p^{a_3, \theta_1}_2 -\frac{1}{8}) +  \mu(\theta_3) \pi(a_3|\theta_3)  (p^{a_3, \theta_3}_1 -\frac{1}{8})\\
&\quad \ge   \frac{1}{2}\mu(\theta_1) \pi(a_3|\theta_1)  p^{a_3, \theta_1}_2  +  \frac{1}{2}\mu(\theta_3) \pi(a_3|\theta_3)  p^{a_3, \theta_3}_1 
\end{align*}
which implies that
\begin{equation}\label{optimal_inequ_expecpayment}
 \mu(\theta_1) \pi(a_3|\theta_1)  p^{a_3, \theta_1}_2  +  \mu(\theta_3) \pi(a_3|\theta_3)  p^{a_3, \theta_3}_1  \ge \mu(\theta_1) \pi(a_3|\theta_1)  \frac{1}{4}  +  \mu(\theta_3) \pi(a_3|\theta_3) \frac{1}{4}
\end{equation}

Hence, by (\ref{actiona2optimal}) and (\ref{actiona1signal}), we know that the principal's optimal utility is 
\begin{align*}
&\sum_{s \in \Sigma} \sum_{\theta} \mu(\theta) \pi(s|\theta) \langle F^{\theta}_{a(s)}, r - {p}^{s, \theta} \rangle\\
     = ~&\Big[ \mu(\theta_2) \pi(a_2|\theta_2) \Big] + \Big[ \frac{1}{4}\mu(\theta_1)\pi(a_1|\theta_1) + \frac{1}{2}\mu(\theta_3)\pi(a_1|\theta_3) \Big] \\
    & \quad \quad + \Big[ \mu(\theta_1) \pi(a_3|\theta_1)  (\frac{1}{2} - p^{a_3, \theta_1}_2)  +  \mu(\theta_3) \pi(a_3|\theta_3) (1- p^{a_3, \theta_3}_1) \Big] \\
     = ~&\frac{1}{3} + \Big[ \frac{1}{4}\mu(\theta_1)(1- \pi(a_3|\theta_1)) + \frac{1}{2}\mu(\theta_3)(1 - \pi(a_3|\theta_3)) \Big] \\
    & \quad \quad + \Big[ \mu(\theta_1) \pi(a_3|\theta_1)  (\frac{1}{2} - p^{a_3, \theta_1}_2)  +  \mu(\theta_3) \pi(a_3|\theta_3) (1- p^{a_3, \theta_3}_1) \Big] \\
     = ~&\frac{1}{3} + \frac{1}{3}\times (\frac{1}{4} + \frac{1}{2})  - \frac{1}{4}\mu(\theta_1) \pi(a_3|\theta_1) - \frac{1}{2}\mu(\theta_3) \pi(a_3|\theta_3) + \frac{1}{2} \mu(\theta_1) \pi(a_3|\theta_1)   + \mu(\theta_3) \pi(a_3|\theta_3) \\
    & \quad - \Big( \mu(\theta_1) \pi(a_3|\theta_1)  p^{a_3, \theta_1}_2  +  \mu(\theta_3) \pi(a_3|\theta_3)  p^{a_3, \theta_3}_1 \Big) \\
     \le ~& \frac{1}{3}\times (1+\frac{3}{4}) +  \frac{1}{4}\mu(\theta_1) \pi(a_3|\theta_1) + \frac{1}{2}\mu(\theta_3) \pi(a_3|\theta_3)  - 
    \Big(  \mu(\theta_1) \pi(a_3|\theta_1)  \frac{1}{4}  +  \mu(\theta_3) \pi(a_3|\theta_3) \frac{1}{4} \Big)\\
     = ~&\frac{1}{3}\times (1+\frac{3}{4})  + \frac{1}{4}\times \frac{1}{3} \pi(a_3|\theta_3) \\
     \le ~&\frac{1}{3} \times 2 =\frac{8}{12} < \frac{9}{12}
\end{align*}
where the first inequality is due to (\ref{optimal_inequ_expecpayment}).
\end{itemize}

{\bf Case $\pi(a_2|\theta_2) < 1$.} By (\ref{actiona1signal}), we know that $\pi(a_1|\theta_2) = 0$ holds in the optimal solution where the {\it zero} contract is applied.  Hence, in this case, the principal must send signal $a_3$ with positive probability. 

By the individual rationality for signal $a_3$, we have 
\begin{align*}
    \mu(\theta_1) (1-\pi(a_1|\theta_1)) [\langle F^{\theta_1}_{a_3}, p^{a_3, \theta_1} \rangle -\frac{1}{8}]& + \mu(\theta_2) (1-\pi(a_2|\theta_2)) [\langle F^{\theta_2}_{a_3}, p^{a_3, \theta_2} \rangle  -\frac{1}{8}] \\
    & \quad + \mu(\theta_3) (1-\pi(a_1|\theta_3)) [\langle F^{\theta_3}_{a_3}, p^{a_3, \theta_3} \rangle -\frac{1}{8}] \ge 0
\end{align*}
Hence, we have 
\begin{align*}
    &\mu(\theta_1) (1-\pi(a_1|\theta_1)) \langle F^{\theta_1}_{a_3}, p^{a_3, \theta_1} \rangle + \mu(\theta_2) (1-\pi(a_2|\theta_2)) \langle F^{\theta_2}_{a_3}, p^{a_3, \theta_2} \rangle + \mu(\theta_3) (1-\pi(a_1|\theta_3)) \langle F^{\theta_3}_{a_3},  p^{a_3, \theta_3} \rangle\\
    & \quad \quad \ge \mu(\theta_1) (1-\pi(a_1|\theta_1)) \frac{1}{8} + \mu(\theta_2) (1-\pi(a_2|\theta_2)) \frac{1}{8} + \mu(\theta_3) (1-\pi(a_1|\theta_3)) \frac{1}{8}
\end{align*}
The utility that the principal gets from signal $a_3$ is at most
\begin{align*}
    & \mu(\theta_1) (1-\pi(a_1|\theta_1)) \langle F^{\theta_1}_{a_3}, r-p^{a_3, \theta_1} \rangle + \mu(\theta_2) (1-\pi(a_2|\theta_2)) \langle F^{\theta_2}_{a_3}, r-p^{a_3, \theta_2} \rangle  \\
    & \quad \quad + \mu(\theta_3) (1-\pi(a_1|\theta_3))\langle F^{\theta_3}_{a_3}, r-p^{a_3, \theta_3} \rangle \\
    = ~&  \mu(\theta_1) (1-\pi(a_1|\theta_1)) \langle F^{\theta_1}_{a_3}, r \rangle + \mu(\theta_2) (1-\pi(a_2|\theta_2)) \langle F^{\theta_2}_{a_3}, r \rangle  + \mu(\theta_3) (1-\pi(a_1|\theta_3))\langle F^{\theta_3}_{a_3}, r \rangle\\
    & \quad \quad - \mu(\theta_1) (1-\pi(a_1|\theta_1)) \langle F^{\theta_1}_{a_3}, p^{a_3, \theta_1} \rangle - \mu(\theta_2) (1-\pi(a_2|\theta_2)) \langle F^{\theta_2}_{a_3}, p^{a_3, \theta_2} \rangle - \mu(\theta_3) (1-\pi(a_1|\theta_3))\langle F^{\theta_3}_{a_3}, p^{a_3, \theta_3} \rangle\\
    \le~& \frac{1}{2}\mu(\theta_1) (1-\pi(a_1|\theta_1)) + \mu(\theta_3) (1-\pi(a_1|\theta_3)) \\
    & \quad \quad -  \Big( \mu(\theta_1) (1-\pi(a_1|\theta_1)) \frac{1}{8} + \mu(\theta_2) (1-\pi(a_2|\theta_2)) \frac{1}{8} + \mu(\theta_3) (1-\pi(a_1|\theta_3)) \frac{1}{8} \Big)\\
    = ~& \frac{1}{3}\Big( \frac{3}{8}(1-\pi(a_1|\theta_1))  + \frac{7}{8} (1-\pi(a_1|\theta_3)) - \frac{1}{8}(1-\pi(a_2|\theta_2))   \Big)
\end{align*}
Hence, by (\ref{actiona2optimal}) and (\ref{actiona1signal}), the total utility the principal gets is at most 
\begin{align*}
& \sum_{s \in \Sigma} \sum_{\theta} \mu(\theta) \pi(s|\theta) \langle F^{\theta}_{a(s)}, r - {p}^{s, \theta} \rangle\\
     \le ~&  \mu(\theta_2) \pi(a_2|\theta_2)  +  \frac{1}{4}\mu(\theta_1)\pi(a_1|\theta_1) + \frac{1}{2}\mu(\theta_3)\pi(a_1|\theta_3) \\
    & \quad \quad + \frac{1}{3}\Big( \frac{3}{8}(1-\pi(a_1|\theta_1))  + \frac{7}{8} (1-\pi(a_1|\theta_3)) - \frac{1}{8}(1-\pi(a_2|\theta_2))   \Big)  \\
     = ~& \frac{1}{3}\Big( \pi(a_2|\theta_2) + \frac{1}{4}\pi(a_1|\theta_1) + \frac{1}{2}\pi(a_1|\theta_3) + \frac{3}{8}(1-\pi(a_1|\theta_1))  + \frac{7}{8} (1-\pi(a_1|\theta_3)) - \frac{1}{8}(1-\pi(a_2|\theta_2)) \Big) \\
     = ~&\frac{1}{3}\Big( \frac{10}{8} + \pi(a_2|\theta_2) - \frac{1}{8}(1-\pi(a_2|\theta_2)) - \frac{1}{8}\pi(a_1|\theta_1) -\frac{3}{8}\pi(a_1|\theta_3) \Big)\\
     \le~& \frac{1}{3}\Big( \frac{10}{8} + \pi(a_2|\theta_2) - \frac{1}{8}(1-\pi(a_2|\theta_2)) \Big) \\
    < ~&\frac{9}{12}
\end{align*}
The last inequality is by $\pi(a_2|\theta_2) < 1$. 

Finally, by combining the discussion of the above two cases, we can conclude that the optimal principal utility is strictly less $\frac{9}{12}$, which is a supremum but not achievable.

\subsection{Proof of Lemma~\ref{lemma_epsic}}

    If all the obtained pairs $(z^{s, \theta}, \pi(s|\theta))$ to Problem (\ref{relaxted_programzeqpip}) are regular, then by (\ref{recovingincontract}), we recover a solution satisfying the claim. Next, we focus on solutions with irregular pairs.

We prove the lemma by first reserving some probability mass and then redistributing it to fix the irregular solutions. Let $\eps>0$ be a sufficiently small positive constant. First, we construct a new signaling scheme $\hat{\pi}$ and introduce one additional signal $\bar{s}$, where $\hat{\pi}(\bar{s}|\theta) = \eps$ for all $\theta\in \Theta$. Accordingly, we let $\hat{\pi}(s|\theta) = \pi(s|\theta) - \pi(s|\theta)\eps$ for all $\theta\in \Theta$ and $s\in \langle n \rangle$. It can be verified that $\hat{\pi}$ satisfies the probability constraint and is a signaling scheme. The signal $\bar{s}$ is used to {\it reserve} some probability mass. We  show that the first constraint in Program~(\ref{relaxted_programzeqpip}) is relaxed by $\eps$ for signals $s \in \langle n \rangle$, 
\begin{equation}\label{inequaltiyrsereare}
\begin{aligned}
    &\sum_{\theta} \mu(\theta)  \Big[ \langle F^{\theta}_{s},  z^{s, \theta} \rangle - \hat{\pi}(s|\theta) c_{s} \Big] - \sum_{\theta} \mu(\theta)  \Big[ \langle F^{\theta}_{i}, z^{s, \theta} \rangle - \hat{\pi}(s|\theta) c_{i} \Big] \\
    \ge  & \sum_{\theta} \mu(\theta)   \pi(s|\theta) (c_{s}  - c_{i} ) - \sum_{\theta} \mu(\theta)   \hat{\pi}(s|\theta) (c_{s}  - c_{i} )  \\
    = & \eps \sum_{\theta} \mu(\theta)   \pi(s|\theta) (c_{s}  - c_{i} )\\
    \ge & -\eps
\end{aligned}
\end{equation}
where the first inequality is by that $\pi$ satisfy the IC constraint in Program~(\ref{relaxted_programzeqpip}) and the last inequality is by $c_s-c_i \ge -1$. Similarly, we can verify that the principal in Program~(\ref{relaxted_programzeqpip}) loses $\eps$ utility by only considering signal $s \in \langle n \rangle$, i.e., 
\begin{align}\label{utilitylossatmost}
\sum_{\theta} \mu(\theta) \sum_{s \in \langle n \rangle}  \langle F^{\theta}_{s}, \pi(s|\theta)r - z^{s, \theta}\rangle - \sum_{\theta} \mu(\theta) \sum_{s \in \langle n \rangle}  \langle F^{\theta}_{s}, \hat{\pi}(s|\theta)r - z^{s, \theta}\rangle \le \eps
\end{align}
Next, we perform the {\it redistribute} step. Without loss of generality, we consider one signal $\hat{s} \in \langle n \rangle$. Let $\hat{\Theta} \subseteq \Theta$ be the set of types containing all the irregular pairs for signal $s$. We construct a new scheme $\tilde{\pi}$ such that 1) for $\theta \in \hat{\Theta}$,  $\tilde{\pi}(\hat{s}|\theta) = \delta = \frac{\eps}{n}$,  $\tilde{\pi}(s|\theta) = \hat{\pi}(s|\theta)$ for $s\neq \hat{s}$ and $\tilde{\pi}(\bar{s}|\theta) = \hat{\pi}(\bar{s}|\theta) - \frac{\delta}{n}$ and 2) for  $\theta \notin \hat{\Theta}$, $\tilde{\pi}(s|\theta) = \hat{\pi}(s|\theta)$ for all $s\in \langle n \rangle \cup \{\bar{s}\}$. We verify that the principal's utility under $\tilde{\pi}$ weakly increases over $\hat{\pi}$ by only considering signal $\hat{s} \in \langle n \rangle$, 
\begin{equation}\label{increasinprincipanutilty}
\begin{aligned}
    &\sum_{\theta} \mu(\theta) \sum_{s \in \langle n \rangle}  \langle F^{\theta}_{s}, \tilde{\pi}(s|\theta)r - z^{s, \theta}\rangle  - \sum_{\theta} \mu(\theta) \sum_{s \in \langle n \rangle}  \langle F^{\theta}_{s}, \hat{\pi}(s|\theta)r - z^{s, \theta}\rangle \\
    = & \sum_{\theta} \mu(\theta) (\tilde{\pi}(\hat{s}|\theta) - \hat{\pi}(\hat{s}|\theta)) \langle F^{\theta}_{\hat{s}}, r\rangle \ge 0
\end{aligned}
\end{equation}
Next, we show that the first constraint for signal $\hat{s}$ in (\ref{relaxted_programzeqpip}) are $(\delta+\eps)$-IC under $\tilde{\pi}$,
\begin{align*}
        &\sum_{\theta} \mu(\theta)  \Big[ \langle F^{\theta}_{\hat{s}},  z^{\hat{s}, \theta} \rangle - \tilde{\pi}(\hat{s}|\theta) c_{\hat{s}} \Big] - \sum_{\theta} \mu(\theta)  \Big[ \langle F^{\theta}_{i}, z^{\hat{s}, \theta} \rangle - \tilde{\pi}(\hat{s}|\theta) c_{i} \Big] \\
        \ge & \sum_{\theta} \mu(\theta)   \hat{\pi}(\hat{s}|\theta) (c_{\hat{s}} -c_i)  -\eps - \Big[ \sum_{\theta} \mu(\theta)  \tilde{\pi}(\hat{s}|\theta) (c_{\hat{s}} -c_i) \Big]\\
        = & \sum_{\theta} \mu(\theta)   (\hat{\pi}(\hat{s}|\theta) - \tilde{\pi}(\hat{s}|\theta) ) (c_{\hat{s}} -c_i)  -\eps\\
        \ge & -\delta-\eps
\end{align*}
where the first inequality is by (\ref{inequaltiyrsereare}) and the second inequality is by that $0\ge \hat{\pi}(\hat{s}|\theta) - \tilde{\pi}(\hat{s}|\theta)\ge -\delta$ and $c_{\hat{s}} -c_i\le 1$.

Similarly, by performing the {\it redistribute} step to all other signals $s\in \langle n \rangle$, we know that the principal's utility from the signal set $\langle n \rangle$ keep increasing by (\ref{increasinprincipanutilty}). Therefore, the principal lose utility at most $\eps$ by (\ref{utilitylossatmost}). Now, for all $s\in \langle n \rangle$ and $\theta \in \Theta$,  $\pi(s|\theta)>0$ by our {\it reserve then redistribute} step. Hence, with (\ref{recovingincontract}), we recover the contracts $p$. Finally, for the signal $\bar{s}$, we can obtain the optimal IC contract $p^{\bar{s}, \theta}$ for all $\theta \in \Theta$ by solving an LP, where the optimal action for the agent is denoted as $a^* \in \langle n \rangle$. Note that signal $\bar{s}$ will only give nonnegative utility to the principal since at least the principal can employ all {\it zero} contracts. This again implies that the principal loses at most $\eps$ utility. To make it a direct scheme, we can combine signal $\bar{s}$ with the direct signal $a^*$. With similar arguments as in Lemma \ref{lemma_recommenddirect}, this will not decrease the principal's utility.

Finally, by letting $\xi = \delta + \eps = \frac{n+1}{n}\eps > 0$, we  conclude the proof.

\subsection{Proof of Lemma~\ref{lemma_eps_toicbyloss}}

    Suppose the $\xi$-IC solution for the Problem \ref{original_problem_contract_signal} is $(\pi, P)$. Now, for each signal $s$ and some $\eta>0$, we construct new contracts $\bar{p}^{s, \theta} = (1-\eta){p}^{s, \theta} + \eta r$ for all $\theta \in \Theta$. Under signal $s$ and contracts $\Big\{\bar{p}^{s, \theta}\Big\}_{\theta \in \Theta}$, the agent chooses action $a(s)$ by $a(s) =  \arg\max_k \Big\{ \sum_{\theta} \mu(\theta) \pi(s|\theta) [F^{\theta}_{k} \bar{p}^{s, \theta} -c_{k}]  \Big\}$ with tie-breaking in favor of the principal. 

    By definition, the mechanism is IC under $(\pi, \bar{P}, a)$, by which we have
    \begin{equation}\label{icpbrapia}
    \begin{aligned}
     \sum_{\theta} \mu(\theta) \pi(s|\theta) [F^{\theta}_{a(s)} \bar{p}^{s, \theta} -c_{a(s)}] &= \sum_{\theta} \mu(\theta) \pi(s|\theta) [F^{\theta}_{a(s)} \big( (1-\eta){p}^{s, \theta} + \eta r\big) -c_{a(s)}] \\
     &\ge \sum_{\theta} \mu(\theta) \pi(s|\theta) [F^{\theta}_{s} \bar{p}^{s, \theta} -c_{s}] \\
     & = \sum_{\theta} \mu(\theta) \pi(s|\theta) [F^{\theta}_{s} \big( (1-\eta){p}^{s, \theta} + \eta r\big) -c_{s}] 
    \end{aligned}
    \end{equation}
    Moreover, by $\xi$-IC solution $(\pi, P)$, 
    we have 
    \begin{equation}\label{xiidpior}
    \sum_{\theta} \mu(\theta) \pi(s|\theta) [F^{\theta}_{s} {p}^{s, \theta} -c_{s}] \ge \sum_{\theta} \mu(\theta) \pi(s|\theta) [F^{\theta}_{a(s)} {p}^{s, \theta} -c_{a(s)}] - \xi
    \end{equation}
    Combining the above two inequalities (\ref{icpbrapia}) and (\ref{xiidpior}), we have that 
    \begin{align}
    \sum_{\theta} \mu(\theta) \pi(s|\theta) [F^{\theta}_{a(s)} \big( -\eta{p}^{s, \theta} + \eta r\big)] &\ge \sum_{\theta} \mu(\theta) \pi(s|\theta) [F^{\theta}_{s} \big( -\eta{p}^{s, \theta} + \eta r\big) ] - \xi \notag\\
    \sum_{\theta} \mu(\theta) \pi(s|\theta) [F^{\theta}_{a(s)} \big( r -{p}^{s, \theta} \big)] &\ge \sum_{\theta} \mu(\theta) \pi(s|\theta) [F^{\theta}_{s} \big( r -{p}^{s, \theta} \big) ] - \frac{\xi}{\eta} \label{compasreaswiths}
    \end{align}

    Finally, we have the principal's utility for the new contract $\bar{p}$ as 
    \begin{align*}
        \sum_{\theta} \mu(\theta) \sum_{s \in \Sigma} \pi(s|\theta) F^{\theta}_{a(s)} \cdot (r - \bar{p}^{s, \theta}) & = \sum_{\theta} \mu(\theta) \sum_{s \in \Sigma} \pi(s|\theta) F^{\theta}_{a(s)} \cdot [r - ((1-\eta){p}^{s, \theta} + \eta r)] \\
        & = \sum_{s \in \Sigma} (1-\eta) \big[\sum_{\theta} \mu(\theta) \pi(s|\theta) F^{\theta}_{a(s)} (r - p^{s, \theta}) \big] \\
        &\ge \sum_{s \in \Sigma} (1-\eta) \big[ \sum_{\theta} \mu(\theta) \pi(s|\theta) [F^{\theta}_{s} \big( r -{p}^{s, \theta} \big) ] - \frac{\xi}{\eta} \big]  \\
        & = (1-\eta) \Big( \sum_{s \in \Sigma}   \sum_{\theta} \mu(\theta) \pi(s|\theta) [F^{\theta}_{s} \big( r -{p}^{s, \theta} \big) ] - n \frac{\xi}{\eta} \Big) \\
        & \ge \sum_{s \in \Sigma}   \sum_{\theta} \mu(\theta) \pi(s|\theta) [F^{\theta}_{s} \big( r -{p}^{s, \theta} \big) ] -n\frac{\xi}{\eta} - \eta\\
        & \ge \sum_{s \in \Sigma}   \sum_{\theta} \mu(\theta) \pi(s|\theta) [F^{\theta}_{s} \big( r -{p}^{s, \theta} \big) ] -(n+1)\sqrt{\xi}
    \end{align*}
    where the first inequality is by (\ref{compasreaswiths}) and the last inequality is by setting $\eta = \sqrt{\xi}$. This concludes the proof.

\section{Proof Omitted from Section~\ref{sec:explicitycontractdesing}}

\subsection{Proof of Proposition~\ref{prop:K}}

First, we define $K$-uniform distributions.

\begin{definition}
    ($K$-uniform distribution \cite{castiglioni2023selling,cheng2015mixture}) A distribution $q\in \Delta_n$ is said to be a $K$ uniform distribution if it is an average of $K$ standard basis vectors in $n$-dimensional space.
\end{definition}

Let $K = \frac{\log(2nB/\eps)}{2B^2\eps^4}$ and $\mathcal{K}$ be the set of all $K$-uniform distributions over state $\Theta$. Given signal $s^*$, denote the induced posterior over $\Theta$ as $q^* \triangleq \prob{\cdot|s^*}$.
Moreover, under $q^*$, the optimal linear contract is $p^* \triangleq p(q^*)$ and the agent's best-response action is $a^*\triangleq a(q^*)$. 

Let $\tilde{q}$ be the empirical distribution of $K$ i.i.d. samples from $q^*$. By definition, $\tilde{q} \in \mathcal{K}$. Further note $E[\tilde{q}] =q^*$, which implies that there exists some distribution $\gamma$ over $\mathcal{K}$ such that $\sum_{\tilde{q} \in \mathcal{K}} \gamma_{\tilde{q}} \tilde{q} = q^*$. Define Let $t_i \triangleq \sum_\theta q^* (\theta) \langle F^{\theta}_i, p^* \rangle$ and $\tilde{t}_i \triangleq \sum_\theta \tilde{q} (\theta) \langle F^{\theta}_i, p^* \rangle$ for all action $i \in \langle n\rangle$, where $t, \tilde{t} \in [0, B]$. Given some $\delta = \eps^2 > 0$, we define $\mathcal{Q} \subseteq \mathcal{K}$ such that $\mathcal{Q} = \{\tilde{q} ~\vert ~\forall i, |t_i-\tilde{t}_i| \le \delta \}$, i.e., the induced utility by $\tilde{q} \in \mathcal{Q}$ is close to that by $q^*$ for all actions. 

By Hoeffding's inequality, we have for all $i \in \langle n\rangle$,
 \[
 \prob{|t_i-\tilde{t}_i| > \delta} \le \frac{\eps}{nB}.
 \]
This implies that in the distribution $\gamma$, at least $\sum_{\tilde{q} \in \mathcal{Q}} \gamma_{\tilde{q}} \ge 1-\frac{\eps}{B}$. In other words, the principal loses at most $\eps$ utility in the decomposition step, i.e., 
\begin{equation}\label{epsdocomposoitionloeseeps}
\begin{aligned}
\sum_{\tilde{q} \in \mathcal{Q}} \gamma_{\tilde{q}} \sum_{\theta} \tilde{q}(\theta) \langle F^{\theta}_{a^*}, r -  p^*\rangle & = \sum_{\tilde{q} \in \mathcal{K}} \gamma_{\tilde{q}} \sum_{\theta} \tilde{q}(\theta) \langle F^{\theta}_{a^*}, r -  p^*\rangle - \sum_{\tilde{q} \in \mathcal{K}\setminus \mathcal{Q}} \gamma_{\tilde{q}} \sum_{\theta} \tilde{q}(\theta) \langle F^{\theta}_{a^*}, r -  p^*\rangle \\
&\ge \sum_{\tilde{q} \in \mathcal{K}} \gamma_{\tilde{q}} \sum_{\theta} \tilde{q}(\theta) \langle F^{\theta}_{a^*}, r -  p^*\rangle  - \sum_{\tilde{q} \in \mathcal{K}\setminus \mathcal{Q}} \gamma_{\tilde{q}} \cdot B \\
&\ge \sum_{\tilde{q} \in \mathcal{K}} \gamma_{\tilde{q}} \sum_{\theta} \tilde{q}(\theta) \langle F^{\theta}_{a^*}, r -  p^*\rangle  - \eps
\end{aligned}
\end{equation}

Next, we examine the IC constraints for $\tilde{q} \in \mathcal{Q}$. By definition of $\mathcal{Q}$, we have
  \begin{align*}
      \sum_\theta \tilde{q}(\theta) \langle F^{\theta}_{a^*}, p^* \rangle- c_{a^*}& \ge \sum_\theta q^*(\theta) \langle F^{\theta}_{a^*}, p^* \rangle- c_{a^*} -\delta \\
      &\ge \sum_\theta q^*(\theta) \langle F^{\theta}_{i}, p^*\rangle- c_{i} -\delta \tag*{($a^*$ is the best-response action)} \\
      &\ge \sum_\theta \tilde{q}(\theta) \langle F^{\theta}_{i}, p^*\rangle - c_{i} -2\delta, \quad \forall i\in \langle n \rangle
 \end{align*}
This implies that under the close posterior $\tilde{q}$,  the tuple $(p^*, a^*)$ is $2\delta$-IC. Furthermore, by a similar argument as Lemma \ref{lemma_eps_toicbyloss}, we can convert it to an IC mechanism by losing some small utility.
\begin{claim}\label{lineqepsictoic}
    For a constant $\eta \in (0, 1)$, we can convert a $2\delta$-IC to an IC mechanism by losing the principal at most $3\sqrt{\delta}$ utility.
\end{claim}
\begin{proof}
    Given the posterior $\tilde{q}$, construct a new contract $\bar{p} = (1-\eta)p^* + \eta r$, and let $\bar{a} = a (\bar{p})$ be the best-response action under contract $\bar{p}$. We have
    \[
    \sum_\theta \tilde{q}(\theta) \langle F^{\theta}_{\bar{a}}, (1-\eta)p^* + \eta r \rangle - c_{\bar{a}} \ge \sum_\theta \tilde{q}(\theta) \langle F^{\theta}_{a^*}, (1-\eta)p^* + \eta r  \rangle - c_{a^*}
    \]
    and by the $2\delta$-IC constraint regarding action $a^*$ under $p^*$, we have
    \[
    \sum_\theta \tilde{q}(\theta) \langle F^{\theta}_{a^*}, p^*\rangle - c_{a^*} \ge \sum_\theta \tilde{q}(\theta) \langle F^{\theta}_{\bar{a}}, p^*\rangle - c_{\bar{a}} -2\delta
    \]
    Hence, by combining the above two inequalities, 
    \[
    \sum_\theta \tilde{q}(\theta) \langle F^{\theta}_{\bar{a}}, -\eta p^* + \eta r \rangle \ge \sum_\theta \tilde{q}(\theta) \langle F^{\theta}_{a^*}, -\eta p^* + \eta r \rangle  -2\delta
    \]
    which implies that 
    \[
    \sum_\theta \tilde{q}(\theta) \langle F^{\theta}_{\bar{a}},   r-p^* \rangle \ge \sum_\theta \tilde{q}(\theta) \langle F^{\theta}_{a^*}, r -p^*\rangle  -\frac{2\delta}{\eta}.
    \]
    Hence, the principal's utility is 
    \begin{align*}
        \sum_\theta \tilde{q}(\theta)  \langle F^{\theta}_{\bar{a}}, r -\bar{p}\rangle &= \sum_\theta \tilde{q}(\theta) \langle F^{\theta}_{\bar{a}}, (1-\eta)(r-p^*) \rangle  \\
        &\ge(1-\eta) \Big(\sum_\theta \tilde{q}(\theta) \langle F^{\theta}_{a^*},  r-p^* \rangle  -\frac{2\delta}{\eta}\Big)\\
        &\ge \sum_\theta \tilde{q}(\theta) \langle F^{\theta}_{a^*},  r-p^* \rangle -\frac{2\delta}{\eta} -\eta\\
        &\ge \sum_\theta \tilde{q}(\theta) \langle F^{\theta}_{a^*},  r-p^* \rangle-3\sqrt{\delta}
    \end{align*}
    where the last inequality is by setting $\eta=\sqrt{\delta}$.
\end{proof}

For each signal $\tilde{q} \in \mathcal{K}$, we solve the optimal contract $p(\tilde{q})$ and action $a(\tilde{q})$. Hence, we have that for distribution $\gamma$,
\begin{align*}
    \sum_{\tilde{q} \in \mathcal{K}} \gamma_{\tilde{q}} \sum_{\theta} \tilde{q}(\theta) \langle F_{a(\tilde{q})}^\theta, r - p(\tilde{q})\rangle& \ge \sum_{\tilde{q} \in \mathcal{Q}} \gamma_{\tilde{q}} \sum_{\theta} \tilde{q}(\theta) \langle F_{a(\tilde{q})}^\theta, r - p(\tilde{q})\rangle \\
    & \ge \sum_{\tilde{q} \in \mathcal{Q}} \gamma_{\tilde{q}} \sum_{\theta} \tilde{q}(\theta) \langle F_{a^*}^\theta, r - p^*\rangle -3\sqrt{\delta} \tag*{(Claim \ref{lineqepsictoic})} \\
    & \ge \sum_{\tilde{q} \in \mathcal{K}} \gamma_{\tilde{q}} \sum_{\theta} \tilde{q}(\theta)  \langle F_{a^*}^\theta, r-p^*\rangle -3\sqrt{\delta} -\eps \tag*{(Equation (\ref{epsdocomposoitionloeseeps}))}\\
    &= \sum_{\theta} q^*(\theta) \langle F_{a^*}^\theta, r - p^*\rangle -4\eps
\end{align*}
Thus, we can restrict our attention to $|\mathcal{K}| = \poly(n^{\frac{\log(2nB/\eps)}{2B^2\eps^4}})$ posteriors.  Then, we can compute a solution $4\eps$-additive approximation by solving the following linear program,
\begin{equation}
    \begin{aligned}
        \max_{\gamma} &\quad   \sum_{q \in \mathcal{K}} \gamma_{q} \sum_{\theta} q(\theta)  \langle F^{\theta}_{a(q)}, r-p(q) \rangle \\
        \textnormal{subject to}&  \quad \mu(\theta) = \sum_{q \in \mathcal{K}} \gamma_{q} q(\theta), \quad \forall \theta \in \Theta
    \end{aligned}
\end{equation}
where the optimal contract $p(q)$ and the best-response action $a(q)$ are known given posterior $q$. In total, there are $O (\poly(n^{\frac{\log(2nB/\eps)}{2B^2\eps^4}}))$ parameters and $O(|\Theta|)$ constraints. Thus, the algorithm runs in $\poly(|\mathcal{K}|, |\Theta|)$ time.

\section{Proofs Omitted from Section~\ref{sec:single}}

\subsection{Proof of Proposition~\ref{singlcontractdirectmech}}

    As the first step, we show that given any feasible solution $(\pi^*, p^*)$, there exists one direct signaling scheme that achieves the same principal utility. Suppose there exist two signals $s_1$ and $s_2$ that induce the agent to take the same action $a$, then the following  hold,
\begin{gather*}
    \sum_{\theta} \mu(\theta)\pi^*(s_1|\theta) [\langle F^{\theta}_{a}, p^* \rangle -c_{a}] \ge \sum_{\theta} \mu(\theta)\pi^*(s_1|\theta) [\langle F^{\theta}_{i}, p^*\rangle -c_{i}], \forall i \in \langle n \rangle \\ 
    \sum_{\theta} \mu(\theta)\pi^*(s_2|\theta) [\langle F^{\theta}_{a}, p^* \rangle  -c_{a}] \ge \sum_{\theta} \mu(\theta)\pi^*(s_2|\theta) [\langle F^{\theta}_{i}, p^* \rangle -c_{i}], \forall i \in \langle n \rangle 
\end{gather*}
Then, in the new scheme $\pi$ with a signal $s$, $\pi (s|\theta) = \pi^*(s_1|\theta) + \pi^*(s_2|\theta)$ defined by combining $s_1$ and $s_2$, we have 
\begin{equation*}
     \sum_{\theta} \mu(\theta)[\pi^*(s_1|\theta) + \pi^*(s_2|\theta)] [\langle F^{\theta}_{a}, p^* \rangle -c_{a}] \ge \sum_{\theta} \mu(\theta)[\pi^*(s_1|\theta) + \pi^*(s_2|\theta)] [\langle F^{\theta}_{i}, p^* \rangle -c_{i}]
\end{equation*}
This implies that the best-response action does not change in the new scheme.  Moreover, the utility of the principal does not change,
\begin{align*}
    \sum_{\theta} \mu(\theta) \pi(s_1|\theta) \langle  F^{\theta}_{a}, r - p^* \rangle + \sum_{\theta} \mu(\theta) \pi(s_2|\theta)\langle F^{\theta}_{a}, r - p^* \rangle = \sum_{\theta} \mu(\theta) [\pi(s_1|\theta) + \pi(s_2|\theta)] \langle F^{\theta}_{a}, r - p^* \rangle 
\end{align*}
Hence, it is without loss of generality to  consider a direct mechanism in the single contract case. Next, we show that such an optimal direct mechanism exists.

 Consider the function $g(p)$ that for any $p \in \mathbb{R}_+$ such that the principal's utility is maximized,
\begin{equation}\label{eq:USC}
    g(p) \triangleq \max_{\pi} \quad  \sum_{\theta} \mu(\theta) \sum_{s \in \Sigma} \pi(s|\theta) \langle F^{\theta}_{a(s)}, r - p \rangle \quad \textnormal{subject to} ~ \pi ~ \textnormal{is feasible},
\end{equation}
whereby the above analysis, we can restrict to a direct scheme with $ a (s)$ as the best response to signal $s$ under signaling scheme $\pi$. Note that given $p$, the optimal scheme $\pi$ must exist since (\ref{eq:USC}) is a linear program.
For any $\pi$, since ties are broken in favor of the principal, the function $\sum_{\theta} \mu(\theta) \sum_{s \in \Sigma} \pi(s|\theta) \langle F^{\theta}_{a(s)}, r - p \rangle$ is upper-semicontinuous in $p$. Hence, since the maximum of semicontinuos function is upper-semicontinuos also \eqref{eq:USC} is upper-semicontinuos.

Then, we show that we can restrict to a bounded and closed set on which we  optimize $p$.
In particular, we can restrict to $0 \le p_\omega \le  \frac{2}{\max_{i} \sum_{\theta}\mu(\theta) F^\theta_{i, \omega}}$ for every $\omega \in \Omega$. Let $i({\omega})\in \arg\max_{i} \sum_{\theta}\mu(\theta) F^\theta_{i, {\omega}}$ for every $\omega \in \Omega$.  Notice that $\sum_{\theta}\mu(\theta) F^\theta_{i({\omega}), {\omega}}>0$ since otherwise the outcome ${\omega}$ cannot be induced and we can remove it.
Indeed, consider a contract $p$ with some $\omega^*$ such that  $p_{\omega^*} >  \frac{2}{ \sum_{\theta}\mu(\theta) F^\theta_{i(\omega^*), \omega^*}}$. 
It is easy to see that by the IC constraints, the agent's total utility is at least the one he obtains from ignoring the signal and always playing action $i({\omega^*})$. Hence, the total utility is at least 
\[ \sum_{\theta} \mu(\theta) [\langle F^{\theta}_{i({\omega^*})}, p \rangle -c_{i({\omega^*})}] \ge  p_{\omega^*} \sum_{\theta}\mu(\theta) F^\theta_{i({\omega^*}), {\omega^*}}-c_{i(\omega^*)} >1.\]
This implies that the expected payment is strictly larger than $1$ and hence the principal's utility is negative.
Hence, we have shown that the optimal mechanism can be obtained from the maximization of an upper-semicontinous function on a  closed and bounded set. Hence, an optimum exists.

\end{document}